\documentclass[journal]{IEEEtran}
\ifCLASSINFOpdf
\else
\fi
\usepackage[utf8]{inputenc}
\usepackage{graphicx,float}
\usepackage[makeroom]{cancel}
\usepackage[dvipsnames]{xcolor}
\usepackage{amsmath,amsfonts,amssymb,mathtools,amsthm}
\usepackage{caption}
\usepackage{algorithm,algorithmic}
\usepackage{subcaption}
\usepackage[most]{tcolorbox}
\usepackage{booktabs}       
\usepackage{hyperref}
\usepackage{multirow}
\usepackage[figuresright]{rotating}

\newtheorem{theorem}{Theorem}
\usepackage{array}
\usepackage{comment}

\newcommand{\bGamma}{\mathbf{\Gamma}}
\newcommand{\bDelta}{\mathbf{\Delta}}
\newcommand{\bH}{\mathbf{H}}
\newcommand{\bQ}{\mathbf{Q}}
\newcommand{\bR}{\mathbf{R}}
\newcommand{\bC}{\mathbf{C}}
\newcommand{\bq}{\mathbf{q}}
\newcommand{\br}{\mathbf{r}}
\newcommand{\bS}{\mathbf{S}}

\newcommand{\bI}{\mathbf{I}}

\newcommand{\by}{\mathbf{y}}
\newcommand{\bx}{\mathbf{x}}

\newcommand{\bz}{\mathbf{z}}
\newcommand{\bzero}{\mathbf{0}}

\newcommand{\bphi}{\boldsymbol{\phi}}
\newcommand{\bmu}{\boldsymbol{\mu}}

\newcommand{\bff}{\mathbf{f}}

\newcommand{\bg}{\mathbf{g}}
\newcommand{\bs}{\mathbf{s}}

\newcommand{\bG}{\mathbf{G}}

\newcommand{\bV}{\mathbf{V}}

\newcommand{\vecc}{\mathbf{c}}
\newcommand{\bsigma}{\boldsymbol{\sigma}}
\newcommand{\bSigma}{\boldsymbol{\Sigma}}
\newcommand{\bLambda}{\boldsymbol{\Lambda}}
\newcommand{\bPhi}{\boldsymbol{\Phi}}
\DeclareMathOperator{\Tr}{Tr}

\definecolor{othercol}{RGB}{242, 243, 249}

\begin{document}

\title{{A Gaussian Sum Filter for Unifying Gaussian and Particle Filters}}

\author{Kostas~Tsampourakis,
and~Víctor~Elvira,~\IEEEmembership{Senior Member,~IEEE}%
\thanks{Kostas Tsampourakis is with the School of Mathematics at the University of Edinburgh, Edinburgh, EH9 3FD, U.K (e-mail: kostas.tsampourakis@ed.ac.uk).}%
\thanks{Víctor Elvira is with the School of Mathematics at the University of Edinburgh, Edinburgh, EH9 3FD, U.K (e-mail: victor.elvira@ed.ac.uk).}
\thanks{V.E. acknowledges support from ARL/ARO under grant W911NF-22-1-0235.}}

\maketitle

\begin{abstract}
State-space models (SSMs) are a broad class of probabilistic models for dynamical systems with many applications in engineering and science. Bayesian filtering is analytically tractable only in the linear–Gaussian setting, where the Kalman filter yields exact posterior distributions. For nonlinear or non-Gaussian SSMs, approximations are required. Two prominent families of approximate methods are Gaussian sum filters (GSFs), which rely on local Gaussian approximations and numerical integration schemes, and particle filters (PFs), which use sequential Monte Carlo sampling. Despite their success, GSFs can suffer from numerical instabilities and severe failures in strongly nonlinear regimes, while PFs are flexible and robust but often demand substantial computational resources to achieve accurate estimates. In this work, we propose the Augmented Gaussian Sum Filter (AGSF), a novel filtering framework that unifies GSFs and PFs through an augmented Gaussian approximation parameterized by latent variables and tunable covariance parameters. By adjusting these covariances, the AGSF interpolates continuously between GSF-like and PF-like behavior, recovering both as special cases. Building on this view, we develop an adaptive AGSF that automatically shifts its behavior according to the local nature of the nonlinearities, acting more like a GSF when Gaussian approximations are reliable and more like a PF when they are not. In a target-tracking application, we demonstrate that AGSF is efficient and robust to common failure modes of both GSFs and PFs. Finally, we empirically validate the switching behavior of the adaptive mechanism in a toy example. 

\end{abstract}


%
\IEEEpeerreviewmaketitle

\section{Introduction}
%
%
%
%

State-space models (SSMs) provide a general probabilistic framework for modeling the dynamics of time-varying systems and their measurement processes concurrently. They introduce hidden states for the system of interest, which are accessed through a noisy observation process. SSMs find countless applications in disciplines of science and engineering such as neuroscience \cite{linderman2019hierarchical}, ecology \cite{buckland2004state}, and time-series analysis \cite{aoki2013state}, among many others.

A problem of interest is to sequentially infer the hidden states given a record of observations, a problem known as state estimation or filtering. In the Bayesian setting, this problem is equivalent to obtaining the state posteriors, also known as the filtering distributions \cite{sarkka2023bayesian}. The linear Gaussian SSM (LGSSM) is a SSM in which dynamics and observation processes are linear and all noises are Gaussian. In the LGSSM, the filtering distributions are analytically tractable and are obtained by the celebrated Kalman filter \cite{kalman1960new, kalman1961new}. 

When the models are nonlinear or the noises are non-Gaussian, the problem is intractable and numerical approximations are needed. One of the earliest solutions to this problem is the so-called \emph{extended Kalman filter}\footnote{It is historically interesting to note that the EKF was extensively used during the Apollo program. In fact EKF-enabled navigation was essential for guidance during Apollo 11’s lunar landing \url{https://www.nasa.gov/wp-content/uploads/2015/04/techbul_20-03-nav_filter_042920.pdf}} (EKF) which works by linearizing the model and makes a Gaussian assumption for the filtering distributions \cite{smith1962application}. The EKF belongs to a larger class of popular algorithms known as \emph{Gaussian filters} which also includes the \emph{unscented Kalman filter} (UKF) \cite{wan2000unscented}, the \emph{quadrature} and \emph{cubature} Kalman filters \cite{4266868, 855552, arasaratnam2009cubature}. Gaussian filters assume a Gaussian filtering distributions and use moment-matching approximations to obtain the sequence of posteriors \cite[Chapter 8]{sarkka2023bayesian}. In a family of algorithms known as \textit{Gaussian sum filters} (GSF) \cite{1100034}, a set of Gaussian filters is used in parallel to form a Gaussian mixture approximation of the filtering distribution. While GSFs is a very expressive class of algorithms, which in the limit of infinite components can converge to any probability density function \cite{SORENSON1971465}, they can perform quite poorly in practice \cite{7524032}. This is due to the properties of the linear approximations that underpin GSFs. Due to the local assumption of the linear approximations, the performance can be significantly hindered as the Gaussian covariances become larger or as the system becomes hight nonlinear, or underlying functions are  often leading to divergent numerical instabilities.

Particle filtering (PF) is an alternative approach which has been used widely for over two decades \cite{gordon1993novel, doucet2001sequential}. Particle filters are sequential Monte-Carlo methods \cite{chopin2020introduction} which approximate the filtering distribution using an ensemble of particles that is propagated and weighted using importance sampling \cite{elviraGMIS, advancesinIS}. They are a flexible class of algorithms which make few assumptions on the nature of the SSM and are relatively easy to implement \cite{doucet2009tutorial}. Moreover they enjoy strong theoretical convergence properties \cite{crisan2002survey}. This has earned them success in multiple applications such as robotics \cite{thrun2002particle}, finance \cite{lopes2011particle}, and tracking \cite{gustafsson2002particle} among others. While particle filters offer flexibility for filtering in nonlinear and non-Gaussian models, their main drawback is computational complexity because of the high number of particles needed due to the curse of dimensionality, which often makes them impractical \cite{daum2003curse}. 

Balancing accuracy and efficiency becomes a crucial consideration when employing Bayesian filters, since we are typically interested not only in state estimation but also in propagating uncertainty through the nonlinear dynamics and measurement updates. In this respect, GSFs and PFs sit at opposite ends of a continuum of approximations: GSFs represent uncertainty with a continuous parametric form, a Gaussian mixture whose components each carry location--scale information (means and covariances), whereas PFs represent uncertainty nonparametrically via a weighted ensemble of point-mass particles, so that distributional shape is encoded only collectively by the cloud of samples and their weights. While GSFs exploit the expressiveness of Gaussian mixtures, their accuracy depends on the nonlinearity of the model over the support of each component, with smaller component covariances generally reducing the impact of nonlinear effects \cite{1100034}. PFs, in contrast, make minimal assumptions and are asymptotically consistent as the number of particles increases, but in practice they can suffer from weight degeneracy, where only a few particles carry most of the weight, necessitating resampling and large particle counts. This degeneracy becomes dramatic as the state dimension grows, reflecting the curse of dimensionality and often making the use of PFs impractical \cite{daum2003curse}. These drawbacks motivate the search for methods that better balance the tradeoff between accuracy and robustness on the one hand and computational cost on the other.

\noindent \textbf{Related work:} The need to keep Gaussian component covariances small in order to have good performance of GSFs was recognized as early as \cite{1100034}, although no explicit mechanism was proposed to exploit it. Early works that propose such mechanisms are \cite{rauh2009nonlinear} and \cite{6104359}.

The mixture Kalman filter (also known as the Rao-Blackwellized particle filter) was introduced in \cite{chen2000mixture} as a sequential Monte Carlo method tailored to conditionally linear Gaussian models, where one samples latent indicator variables while marginalizing the conditionally linear Gaussian substructure with a Kalman filter, yielding substantial variance reduction relative to particle methods applied directly to the full state. This work although conceptually similar, is effectively designed to solve a different problem, namely that of filtering when the hidden state possesses both continuous and discrete components. 

In the signal processing literature, a complementary construction has been developped \cite{kotecha2003gaussian}. Starting from a Gaussian particle filter \cite{1232326} that maintains a single Gaussian approximation but updates its moments by propagating particles through the nonlinear dynamics/measurement maps (avoiding linearization), they build Gaussian sum particle filters as bank of such Gaussian particle filters to obtain a weighted Gaussian sum approximation of the predictive and filtering distributions; they also describe an extension to additive non-Gaussian noise by approximating the noise as a Gaussian mixture, yielding a bank of Gaussian-noise models. 

A more application-driven line of work, in aerospace, tracking, and robotics, has focused on actively controlling component covariances through adaptive splitting/merging. In particular, an entropy-based nonlinearity detection to trigger component splitting during uncertainty propagation in a Gaussian-mixture framework is proposed \cite{demars2013entropy}.  In \cite{faubel2009split}, a split-and-merge unscented Gaussian sum filter with an explicit mechanism intended to avoid unnecessary splitting in locally linear regions is presented. In robotics-oriented tracking, \cite{havlak2013discrete} discusses a sigma-point-based splitting approach motivated by nonlinear dynamics. For nonlinear measurement updates, \cite{tuggle2018automated} develop an automated splitting procedure targeted specifically at the update step. Finally, recent efforts aim to impose explicit covariance bounds during mixture re-approximation via linear-matrix-inequality inequalities, trading additional optimization complexity for direct control of mixand spread \cite{psiaki2015gaussian, 7524032}.

\noindent \textbf{Contributions:} 
In this work,\footnote{A limited version of this work was presented by the authors in the conference paper \cite{10095899}. {Compared to the conference paper, here we derive the algorithm for a general Gaussian moment matching approximation. We derive a criterion for specifying the augmentation covariance, and we prove that the AGSF algorithm interpolates between the GSF and the BPF. We have introduced an adaptive version of the AGSF. We expanded the experiments, showing improvement to both the GSF and BPF, as well as showcasing the adaptive capabilities of the AGSF.}} we introduce the \textit{augmented Gaussian sum filter} (AGSF), a new framework which addresses the drawbacks of GSFs and PFs while combining features of both by taking an innovative approach away from existing strategies. The key idea is to use a Gaussian convolution identity to express a single Gaussian component as a mixture of narrower components whose covariances can be set directly as free parameters, in a process that we call \emph{augmentation}. By adjusting these augmentation covariances, the AGSF continuously trades off local-approximation bias (which grows when a component becomes too diffuse under strong nonlinearities) against Monte Carlo variance (which dominates when the component is represented purely by point particles).

This approach differs from existing covariance-control strategies in the literature, which typically enforce narrowness indirectly via split/merge or heuristic nonlinearity detectors, and from covariance-bounding schemes that rely on solving linear matrix inequality (LMI) programs to impose explicit upper bounds on the covariance \cite{7524032, psiaki2015gaussian}. In our approach, the Gaussian splitting is a natural outcome of the augmentation, which easily ensures the positive definiteness of resulting covariances without having to impose complex inequalities, while allowing for maximal flexibility in the selection of component covariances. It also differs from particle–Gaussian hybridizations of \cite{chen2000mixture} and \cite{crisan2015generalised} as well as the Gaussian sum particle filter of \cite{kotecha2003gaussian}. These approaches blend features of Gaussian and particle filters in a very different way to the AGSF, which interpolates between GSFs and PFs that are shown to be special cases within the framework. 

The main contributions are summarized as follows:
\begin{itemize}
    \item {We propose a novel splitting scheme by reinterpreting the well known Gaussian convolution identity. The scheme allows us to express a Gaussian distribution as a Gaussian mixture with smaller components, and do so in an efficient and flexible manner.}
    \item We propose a novel class of filtering algorithms called augmented Gaussian sum filters (AGSF) which unifies the classes of Gaussian and particle filters. The AGSF uses a novel augmented Gaussian approximation method based on a Gaussian integral identity.
    \item We derive an adaptive version of the AGSF which can behave like a particle filter, a Gaussian filter, or something in between, according to the local features of the nonlinearities. The adaptive version uses a novel optimization problem which is used to automatically select the augmentation covariance. 
    \item We present simulation results showing the robustness of the AGSF compared to the bootstrap particle filter (BPF) and GSFs. Our results support the view of that the AGSF as interpolate between GSFs and PFs. We also demonstrate the adaptive capability of the AGSF in a model with mixed linear-nonlinear regimes.
\end{itemize}

The rest of the paper is structured as follows. Section \ref{sec 2} summarizes Bayesian filtering for state-space models and the basic GSF and BPF algorithms. In Section \ref{sec 3}, we introduce the basic Gaussian augmentation and the augmented Gaussian approximations that it underpins. In Section \ref{sec 3.D}, we introduce the optimization problem that is used to select the covariances automatically. In Section \ref{sec 4}, we derive the novel AGSF algorithm and prove that it unifies the GSF and the BPF. Finally in Section \ref{sec 5}, we present our experimental results comparing the AGSF to GSF and BPF.

\section{Background}\label{sec 2}
{
\subsection{Notation}
We denote by $\mathcal{N}(\cdot|\bmu, \bSigma)$ the multivariate Gaussian pdf with mean $\bmu$ and covariance $\bSigma$. For any function $\bphi(\bx)$, the notation $\mathcal{I}[\cdot ; \bz, \bDelta]$ denotes the Gaussian integral
\begin{equation*}
    \mathcal{I}[\bphi(\bx) ; \bz, \bDelta] = \int \bphi(\bx) \mathcal{N}(\bx | \bz, \bDelta) d\bx.
\end{equation*}
We denote by $\bz_{mn}$ and $\bs_{mnl}$ the auxiliary variables introduced in the prediction and update steps of the AGSF algorithm, respectively. We also denote by  $\bDelta_m^{(t)}$ and $\bLambda_{mn}^{(t)}$ the augmentation covariances introduced at time $t$, for the prediction and update steps of the AGSF, respectively. Finally, we denote $a\wedge b = \min(a,b)$.
}

\subsection{Bayesian Filtering}
Consider an additive, nonlinear state-space model (SSM), defined by
\begin{align}\label{ssm1}
    \bx_t & = \bff(\bx_{t-1}) + \bq_t, \\ \label{ssm2}
    \by_t & = \bg(\bx_{t}) + \br_t,
\end{align}
where $\bx_t\in \mathbb{R}^{d_x}$ and $ \by_t\in \mathbb{R}^{d_y}$ are the state and observation vectors, $\bff$ and $\bg$ are the nonlinear dynamics and observation functions, while $\bq_t \sim \mathcal{N}(\bzero, \bQ) $ and $\br_t \sim \mathcal{N}(\bzero, \bR)$ are additive noise vectors. The goal of \textit{Bayesian filtering} is to compute the marginal posterior distribution, called \textit{filtering distribution} $p({\bx_t}|\by_{1:t})$ of the state $\bx_t$, given $\by_{1:t} \equiv \{\by_1, \dots, \by_t \}$. The computation can be done recursively, and the filtering distribution at time $t$ can be computed in two steps, given the one at time $t-1$ and the new observation $\by_t$. These steps are (i) the computation of the predictive distribution of the state $\bx_t$ given observations up to time $t-1$ (\textit{prediction step}),
\begin{equation}\label{pred}
    p(\bx_t|\by_{1:t-1}) = \int p(\bx_t|\bx_{t-1})p(\bx_{t-1}|\by_{1:t-1})d\bx_{t-1},
\end{equation}
and (ii) the Bayesian update of this distribution upon receiving the new observation at time $t$ (\textit{update step}):
\begin{equation}\label{upd}
    p(\bx_t|\by_{1:t}) = \frac{p(\by_t|\bx_{t})p(\bx_{t}|\by_{1:t-1})}{\int p(\by_t|\bx_{t})p(\bx_{t}|\by_{1:t-1})d\bx_t}.
\end{equation}
The prediction and update steps can be computed exactly for a linear Gaussian SSM. When the functions are nonlinear, we define approximations of the filtering distribution and compute approximately the prediction and update steps of Eqs. \eqref{pred}-\eqref{upd}.

\subsection{Gaussian moment matching}\label{sec:2.B}
Gaussian moment matching approximations provide a general way for approximating the joint distribution of a pair of random variables by a Gaussian density, {and have been used in various Bayesian filters \cite[Chapter 8]{sarkka2023bayesian}).} Below, we briefly review these approximations for the case of an additive transform. 

\noindent\textbf{Gaussian moment matching of an additive transform.} For the pair of random variables $\bx\in\mathbb{R}^{d}$ and $\by\in\mathbb{R}^{d'}$ defined by
\begin{align} \label{system 1}
    \bx &\sim \mathcal{N}(\cdot|\bmu, \bSigma), \\ \label{system 2}
    \by &= \bff(\bx) + \br,
\end{align}
where $\bff:\mathbb{R}^{d}\rightarrow\mathbb{R}^{d'}$ and $\br\sim \mathcal{N}(\bzero, \bR)$, the exact joint pdf is given by
\begin{align}\label{exact_joint}
    p(\bx, \by) = \mathcal{N}(\by|\bff(\bx), \bR) \mathcal{N}(\bx | \bmu, \bSigma) ,
\end{align}
which is a non-Gaussian density due to the nonlinearity $\bff$. The Gaussian moment matching approximation of the joint is
 \begin{equation}\label{eq:joint gaussian}
    p(\bx, \by ) \approx \mathcal{N}\Bigg( \begin{pmatrix} \bx \\ \by \end{pmatrix} \Big| \begin{pmatrix} \bmu \\   \bmu_{\by} \end{pmatrix}, \begin{pmatrix} \bSigma &   \bC_{\bx\by} \\ \bC_{\bx\by}^T &   \bS_{\by} \end{pmatrix} \Bigg),
\end{equation}
with,
\begin{align} \label{GMM1}
    \bmu_{\by} &= \int \bff(\bx) \mathcal{N}(\bx|\bmu, \bSigma) d\bx, \\ \label{GMM2}
    \bS_{\by}  &= \int (\bff(\bx) - \bmu_{\by})(\bff(\bx) - \bmu_{\by})^T \mathcal{N}(\bx|\bmu, \bSigma) d\bx +\bR, \\ \label{GMM3}
    \bC_{\bx\by} &= \int (\bx - \bmu)(\bff(\bx) - \bmu_{\by})^T \mathcal{N}(\bx|\bmu, \bSigma) d\bx.
\end{align}
Since the integrals \eqref{GMM1}-\eqref{GMM3} are in general intractable, they have to be approximated using some numerical scheme. Note that the approximation is used twice at every iteration of Bayesian filtering: once in the prediction step where we replace $(\bx, \by)$ by $ (\bx_{t-1}, \bx_t)$, and once in the update step where we replace $(\bx, \by)$ by $(\bx_{t}, \by_t)$.

Two of the most popular ways to approximate the moments is to use a Taylor expansion of the nonlinear function $\bff$, leading to the \textit{linear approximation of an additive transform} \cite[Algorithm 7.1]{sarkka2023bayesian}, or the unscented transform, leading to the \textit{unscented approximation of an additive transform} \cite[Algorithm 8.15]{sarkka2023bayesian}. We review these approximations below. 

\noindent\textbf{Linear approximation of an additive transform:} It consists of a linear Taylor expansion of $\bff$ around $\bmu$. Substituting $\bff(\bx)$ in Eq. \eqref{system 2} by its first-order expansion, we obtain
\begin{equation}\label{Taylor moment-matching}
    \by\simeq\bff(\bmu) + \nabla\bff(\bmu)(\bx-\bmu) + \br,
\end{equation}
 which is an affine transformation of $\bx$. Hence, the joint is approximated as in Eq. \eqref{eq:joint gaussian} with
 \begin{align}
     \bmu_{\by} &= \bff(\bmu), \\
     \bS_{\by} &= \nabla \bff(\bmu) \bSigma \nabla \bff(\bmu)^T + \bR, \\
     \bC_{\bx\by} &= \bSigma \nabla \bff(\bmu)^T.
 \end{align} 
A geometric illustration of the linear approximation can be found in Fig. \ref{fig:EKF update}.

\noindent\textbf{Unscented approximation of an additive transform.} Another way to approximate $p(\bx, \by)$ with a Gaussian uses the so-called \textit{unscented transform} (UT) \cite{Julier96ageneral}. The UT deterministically chooses a set of \textit{sigma-points} that captures the first and second order moments of the distribution of $\bx$. The sigma-points are then propagated through the nonlinearity $\bff$, and the resulting points are used to reconstruct the joint Gaussian \cite[Algorithm 8.15]{sarkka2023bayesian}.

Formally, $2d_x+1$ sigma-points are selected as follows,
\begin{align}
    \bsigma^{(0)} &= \bmu, \\
    \bsigma^{(\pm i)} &= \bmu \pm \sqrt{d_x+\lambda} \cdot \bSigma^{1/2}_{\bullet i}, \ i=1,\dots,d_x,
\end{align}
where $\{ \bsigma^{(i)}\}_{i=-d_x}^{d_x}$ are the sigma points,  $\bSigma^{1/2}_{\bullet i}$ is the $i^{th}$ column of the matrix square-root of the covariance $\bSigma$ and $\lambda = \alpha^2(d_x+\kappa) - d_x$, where $\alpha$ and $\kappa$ are parameters that determine the spread of the sigma-points around the mean \cite{wan2001unscented}. The sigma-points are used to construct estimates of the moments of \eqref{exact_joint} as follows:
\begin{align}\label{unscented mean}
    \bmu_{\by}^{\sigma} &= \sum_{i=-d_x}^{d_x} \omega_i \bff(\bsigma^{(i)}), \\ \label{unscented cov}
    \bS_{\by}^{\sigma} &= \sum_{i=-d_x}^{d_x} \tilde \omega_i (\bff(\bsigma^{(i)})- \bmu_{\by}^{\sigma})(\bff(\bsigma^{(i)})- \bmu_{\by}^{\sigma})^T + \bR, \\ \label{unscented cor}
    \bC_{\bx\by}^{\sigma} &= \sum_{i=-d_x}^{d_x} \tilde \omega_i (\bsigma^{(i)}-\bmu)(\bff(\bsigma^{(i)})- \bmu_{\by}^{\sigma})^T,
\end{align}
where
\begin{align}
    \omega_0 &= \frac{\lambda}{d_x+\lambda}, \ \tilde \omega_0 = \frac{\lambda}{d_x+\lambda} + (1-\alpha^2+\beta), \\
    \omega_i &= \tilde \omega_i = \frac{1}{2(d_x+\lambda)}, \ i=\pm1,\dots,\pm d_x,
\end{align}
and where $\beta$ is an additional parameter that can be used to incorporate prior information on the distribution of $\bx$ \cite{wan2001unscented}. Finally, we use the estimates of Eqs. \eqref{unscented mean}-\eqref{unscented cor} to build the Gaussian approximation
 \begin{equation}\label{eq:joint unscented gaussian}
    p(\bx, \by ) \approx \mathcal{N}\Bigg( \begin{pmatrix} \bx \\   \by \end{pmatrix} \Big| \begin{pmatrix} \bmu \\   \bmu_{\by}^{\sigma} \end{pmatrix}, \begin{pmatrix} \bSigma &   \bC_{\bx\by}^{\sigma} \\ (\bC_{\bx\by}^{\sigma})^{T} &   \bS_{\by}^{\sigma} \end{pmatrix} \Bigg).
\end{equation}

\noindent\textbf{Transform for non-additive models.} The Gaussian approximations that we have described can be applied to non-additive transformations of the form
\begin{align}
    \by &= \bff(\bx, \br).
\end{align}
This generalization is straightforward and involves treating the variable $\tilde\bx = (\bx, \br)$ jointly, as the new state variable. For a detailed derivation, see Algorithms 8.2, 7.2 and 8.16 in  \cite{sarkka2023bayesian}, for the generic Gaussian moment-matching approximation, the linear approximation, and the unscented approximation, respectively.

\subsection{Gaussian and Gaussian sum filters}
Gaussian filters use Gaussian moment matching to approximate the joint distribution of $(\bx_{t-1},\bx_t)$ during the prediction step and the joint of $(\bx_{t},\by_t)$ during the update step  \cite{maybeck1982stochastic}. Using different moment matching approximations results in different Gaussian filters, e.g., the \textit{extended Kalman filter} (EKF) which uses the \textit{linear approximation to an additive transform} and the \textit{unscented Kalman filter} (UKF) which uses the \textit{unscented transform}. Other examples include the cubature Kalman filter \cite{arasaratnam2009cubature}, the Gaussian particle filter \cite{1232326}, and Gauss-Hermite Kalman filter \cite{855552}. Gaussian sum filters run in parallel multiple Gaussian filters , such as EKFs or UKFs \cite{SORENSON1971465}.

\subsection{Bootstrap particle filter}\label{sec 2.D}
Particle filters (PFs) are sequential Monte-Carlo methods \cite{doucet2001sequential,chopin2020introduction} that make no assumption about the form of the filtering distribution, and instead use a set of weighted samples (or particles) to represent the filtering distribution. At each time step, they propose a new set of samples using a proposal density and then use importance sampling \cite{robert1999monte} to assign weights to the samples according to the observed data \cite{gordon1993novel, kitagawa1996monte}.

The \textit{bootstrap particle filter} (BPF) \cite[Algorithm 11.9]{sarkka2023bayesian} uses the dynamical model transition density $p(\bx_t|\bx_{t-1})$ as a proposal leading to importance weights proportional to the likelihood. The bootstrap filter also typically performs resampling to avoid particle degeneracy, when an effective sample size (ESS) \cite{elvira2022rethinking} is below a threshold \cite{doucet2009tutorial}. 

\section{Augmentation-based Approximations}\label{sec 3}
 
In this section, we propose a novel \textit{augmented Gaussian approximation} which generalizes the Gaussian moment-matching approximations. The novel procedure is based on a well known Gaussian identity and can be used to make a Gaussian mixture approximation for the joint of a pair of random variables. We derive augmented analogs of the linear and unscented approximations. In Section \ref{sec 4}, we will use the \textit{augmented Gaussian approximation} to derive a novel class of filtering algorithms called \textit{augmented Gaussian sum filters} (AGSF).

\subsection{Augmentation of a Gaussian integral}\label{sec 3.A}
{Here we introduce the augmentation of a Gaussian integral and {use it to propose a novel Gaussian splitting scheme} that is used in the prediction and update steps of the our novel AGSF algorithm.} The augmentation is based on the Gaussian integral identity,
\begin{equation}\label{convolution}
        \mathcal{N} (\bx|\bmu, \bSigma)  = \int \mathcal{N}(\bx|\bGamma \bz  + \vecc, \bDelta)  \mathcal{N} (\bz|\bmu_{\bz}, \bSigma_{\bz})  d\bz,
\end{equation}
 where $\bx \in\mathbb{R}^{d_{x}}$ and $\bz\in\mathbb{R}^{d_z}$. We make a probabilistic interpretation of this identity, treating $\bz$ as a latent random variable on which the variable of interest $\bx$ depends. Eq. \eqref{convolution} requires the following relations to hold:
 \begin{align} \label{augm params 1}
     \bmu &= \bGamma \bmu_{\bz} + \vecc, \\ \label{augm params 2}
     \bSigma &= \bGamma \bSigma_{\bz} \bGamma^T + \bDelta, \\ \label{augm params 3}
     \bDelta &\succeq \bzero, \bSigma_{\bz} \succeq \bzero.
 \end{align}
 Without loss of generality, for the rest of this paper we make the choices $ d_z = d_x, \vecc = \bzero, \bmu_{\bz} = \bmu, \bGamma= \bI$, and $ \bSigma_{\bz} = \bSigma - \bDelta$. We discuss this choice at the end of this section. With this choice the only free parameter to choose is the covariance matrix $\bDelta$, which must satisfy $\bSigma \succeq \bDelta \succeq \bzero$. 

\noindent  {\textbf{Augmentation-based splitting:} We use the augmentation to propose a novel Gaussian splitting scheme}, i.e., to make a Gaussian mixture approximation of $\mathcal{N}(\bx|\bmu, \bSigma)$. We do this by sampling values of the auxiliary variable $\bz$, that is, we approximate the right-hand side of Eq. \eqref{convolution} using Monte Carlo integration as
 
 \begin{equation} \label{A-KDE}
    \mathcal{N} (\bx|\bmu, \bSigma) \simeq \frac{1}{N} \sum_{n=1}^N \mathcal{N}(\bx | \bz_n, \bDelta),
 \end{equation}

 \noindent where $\bz_n \sim \mathcal{N}(\bmu, \bSigma-\bDelta)$. Note that for $\bDelta \to \bzero$, the approximation of Eq. \eqref{A-KDE} is a particle approximation of the original Gaussian, whereas it is easy to verify that for $\bDelta = {\bSigma}   $ we recover the initial Gaussian. Intermediate values of $\bDelta$ can be interpreted as kernel density estimators of \eqref{convolution}.
 
 In the next section, Eq. \eqref{A-KDE} will be used to split the components of the filtering and predictive distributions in the prediction and update steps of the AGSF, respectively.

\subsection{Augmented Gaussian approximations}\label{sec 3.B}
We now combine the augmentation of the Section \ref{sec 3.A} with the Gaussian moment matching approximation to approximate the joint distribution of $(\bx, \by)$ in Eqs. \eqref{system 1}-\eqref{system 2}.  We use the mixture approximation of Eq. \eqref{A-KDE} together with Gaussian moment-matching, to construct a Gaussian mixture approximation of the joint in Eq. \eqref{exact_joint}. We have
\begin{align}
    p(\bx, \by) & = \mathcal{N}(\by|\bff(\bx), \bR) \mathcal{N}(\bx | \bmu, \bSigma) \\ \label{augm-approx}
    &\simeq \mathcal{N}(\by|\bff(\bx), \bR) \frac{1}{N} \sum_{n=1}^N \mathcal{N}(\bx | \bz_n, \bDelta) \\ \label{inter mix}
    &= \frac{1}{N} \sum_{n=1}^N \mathcal{N}(\by|\bff(\bx), \bR) \mathcal{N}(\bx | \bz_n, \bDelta) \\ \label{mixture joint approx}
    & \simeq \frac{1}{N} \sum_{n=1}^N \mathcal{N}\Bigg( \begin{pmatrix} \bx \\   \by \end{pmatrix} \Bigg| \begin{pmatrix} \bz_n \\   \bmu_{\by, n} \end{pmatrix}, \begin{pmatrix} \bDelta &   \bC_{\bx\by,n} \\ \bC_{\bx\by,n}^T &   \bS_{\by, n} \end{pmatrix} \Bigg),
\end{align}
where $\bz_n \sim \mathcal{N}(\bmu, \bSigma - \bDelta)$, and
\begin{align} \label{cond y mean}
    \bmu_{\by, n} &= \int \bff(\bx) \mathcal{N}(\bx | \bz_n, \bDelta) d\bx,\\ \label{cond y cov}
    \bS_{\by, n}  &= \int (\bff(\bx) - \bmu_{\by})(\bff(\bx) - \bmu_{\by})^T \mathcal{N}(\bx|\bz_n, \bDelta) d\bx +\bR, \\ \label{cond cross-cov}
    \bC_{\bx\by,n} &= \int (\bx - \bz_n)(\bff(\bx) - \bmu_{\by})^T \mathcal{N}(\bx|\bz_n, \bDelta) d\bx.
\end{align}

\noindent In Eq. \eqref{augm-approx}, we use the mixture approximation of Eq. \eqref{A-KDE}, whereas in Eq. \eqref{mixture joint approx} we use for each mixture component of Eq. \eqref{inter mix} the \textit{Gaussian moment matching of an additive transform} from Section \ref{sec:2.B}, also \cite[Algorithm 8.1]{sarkka2023bayesian}. To generalize to a non-additive transform we use \cite[Algorithm 8.2]{sarkka2023bayesian} instead. Below, we introduce the augmented analogs of the linear and unscented Gaussian moment-matching approximations by approximating the moments of Eqs. \eqref{cond y mean}-\eqref{cond cross-cov} accordingly.

\subsubsection{Augmented linear approximation} \label{sec 3.B.1}
By applying the \textit{linear approximation of an additive transform} to approximate the moments in Eqs. \eqref{cond y mean}-\eqref{cond cross-cov}, we obtain
\begin{align} \label{L cond y mean}
    \bmu_{\by, n} &= \bff(\bz_n),\\ \label{L cond y cov}
    \bS_{\by,n} &= \nabla\bff(\bz_n)\bDelta \nabla\bff(\bz_n)^T + \bR, \\ \label{L cond cross-cov}
    \bC_{\bx\by,n} &= \bDelta \nabla \bff(\bz_n)^T,
\end{align}

\noindent where $\nabla\bff$ denotes the Jacobian of $\bff$. An illustration of the joint approximation coming from the augmented linear approximation is shown in Fig. \ref{fig:AGSF recursion}. Using the non-additive version of the linear approximation \cite[Algorithm 7.2]{sarkka2023bayesian}, we can generalize the augmented linear approximation to nonlinear models.

\subsubsection{Augmented unscented approximation}\label{augm unscented approx}
Similarly, we can use the unscented approximation to approximate the Eqs. \eqref{cond y mean}-\eqref{cond cross-cov}. In this way, we obtain
\begin{align}\label{aug unscented mean}
    \bmu_{\by,n}^{\sigma} &= \sum_{i=-d_x}^{d_x} \omega_i \bff(\bsigma_{n}^{(i)}), \\ \label{aug unscented cov}
    \bS_{\by,n}^{\sigma} &= \sum_{i=-d_x}^{d_x} \tilde \omega_i (\bff(\bsigma_{n}^{(i)})- \bmu_{\by,n}^{\sigma})(\bff(\bsigma_{n}^{(i)})- \bmu_{\by,n}^{\sigma})^T + \bQ, \\ \label{aug unscented cor}
    \bC_{\bx\by,n}^{\sigma} &= \sum_{i=-d_x}^{d_x} \tilde \omega_i (\bsigma_{n}^{(i)}-\bz_n)(\bff(\bsigma_{n}^{(i)})- \bmu_{\by,n}^{\sigma})^T,
\end{align}

where $\{\bsigma_{n}^{(i)}\}_{i=-d_x}^{d_x}$ denote the sigma-points used to reconstruct $\mathcal{N}(\bz_n, \bDelta)$. Using the non-additive version of the unscented approximation \cite[Algorithm 8.15]{sarkka2023bayesian}, we can generalize the augmented linear approximation to nonadditive models.

\begin{figure}
     \centering
     \begin{subfigure}[b]{0.93\linewidth}
         \centering
         \includegraphics[width=\linewidth]{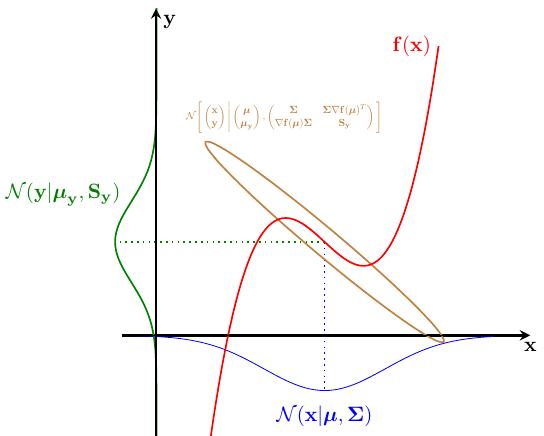}
         \caption{}
         \label{fig:EKF update}
     \end{subfigure}
     \hfill
     \begin{subfigure}[b]{\linewidth}
         \centering
         \includegraphics[width=\linewidth]{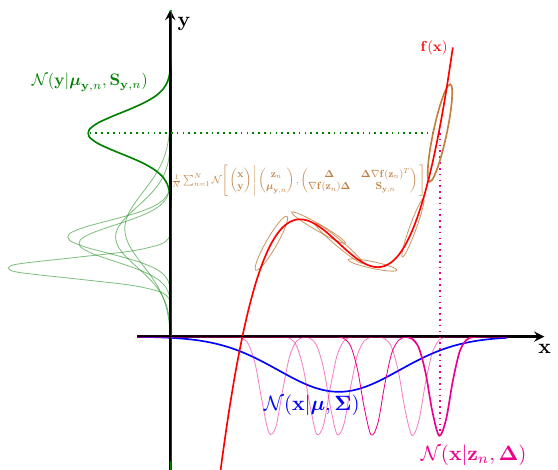}
         \caption{}
         \label{fig:AGSF recursion}
     \end{subfigure}
        \caption{(a) Geometric illustration of the linear approximation of an additive transform. The prior Gaussian distribution for $\bx$ (blue) is propagated through a linear approximation to yield a joint Gaussian approximation (brown) and the marginal of $\by$ (green). (b) Geometric illustration for  the augmented linear approximation. The prior Gaussian (blue), is split into Gaussian components (magenta) which are used for the joint approximation (brown) and the approximation of the marginal of $\by$ (green).}
\end{figure}

\subsection{Selection of parameters}\label{sec 3.C}

We now address the choice of parameters introduced in Section \ref{sec 3.A}. Note that the identity of Eq. \eqref{convolution} holds for any choice of the parameters, $d_z, \vecc, \bGamma, \bDelta, \bmu_{\bz},$ and $\bSigma_{\bz}$ 
that satisfy the constraints \eqref{augm params 1}-\eqref{augm params 3}. However, this set is overparametrized, since if we allow for degenerate 
covariance matrices $\bDelta$ and $\bSigma_{\bz}$, we may restrict $\bz$ to a lower dimensional subspace of $\mathbb{R}^{d_z}$ effectively controlling the dimensionality $d_z$. Moreover, we can absorb $\bGamma$ into $\bSigma_{\bz}$ and $\vecc$ into $\bmu_{\bz}$. Thus, without loss of generality, we make the convenient choices $d_z = d_x, \vecc = \bzero, \bGamma= \bI$. Finally, to satisfy Eqs.\eqref{augm params 1}-\eqref{augm params 3}, we must have $\bmu_{\bz} = \bmu$ and $\bSigma_{\bz} = \bSigma - \bDelta$.

With these choices, the only free parameter of the augmentation is the covariance matrix $\bDelta$. As mentioned in Sec. \ref{sec 3.A}, this matrix allows us to obtain a family of approximations that has as extreme cases a particle approximation (with $\bDelta \to\bzero$), and the original Gaussian (with $\bDelta=\bSigma$). Moreover, in the context of the \textit{augmented Gaussian approximations}, we are able to tune between Monte Carlo approximation ($\bDelta\to\bzero$) and linear or unscented approximations (with $\bDelta=\bSigma$). All choices in our framework that lie in-between these two extremes lead to novel algorithms. 

\subsection{Automatic selection of $\bDelta$}\label{sec 3.D}
In this section, we motivate a procedure for determining the augmentation covariance $\bDelta$ automatically via a semidefinite program. We derive the program in the context of the \textit{augmented linear approximation} (Section \ref{sec 3.B.1}), which allows tractable analytical computations. 

Note that the estimator of variable $\by$ defined in Eq. \eqref{system 2}, is the sample average of the moments given in Eq. \eqref{L cond y mean} and is given by
\begin{align}\label{mean estimator}
    \widehat \bmu_{\by} = \frac{1}{N} \sum_{n=1}^N \bff(\bz_n),
\end{align}
where $\bz_n \sim \mathcal{N}(\bmu, \bSigma - \bDelta)$. The following theorem expresses the MSE of the estimator in Eq. \eqref{mean estimator} as a function of the covariance parameters $\bSigma$ and $\bDelta$, the number of Monte-Carlo samples $N$, the Jacobian, and the Hessian of the nonlinearity $\bff$. It will be exploited later to determine the augmentation covariances automatically.
\begin{theorem}\label{MSE theorem}
Let $\bff:\mathbb{R}^{d_x} \rightarrow \mathbb{R}^{d_y}$ be twice continuously differentiable. Then the MSE of the estimator $\widehat \bmu_{\by}$ defined in Eq. \eqref{mean estimator} is given by
\begin{equation} \label{MSE bound}
    \begin{split}
         MSE(\widehat \bmu_\by) &\simeq \frac{1}{N} \Tr\big((\bSigma-\bDelta)\nabla\bff(\bmu)^T\nabla\bff(\bmu)\big) \\
          &+ \frac{1}{4} \sum_{i=1}^{d_y} \Tr\big( \bDelta \nabla^2\bff_i(\bmu)^2\big),
    \end{split}
\end{equation}
by ignoring terms of higher order in $\bDelta$.
\end{theorem}
\begin{proof}
Given in \nameref{Appendix A}.
\end{proof}
\noindent  The MSE of the estimator is given by two terms, the first term of Eq. \eqref{MSE bound} corresponds to the error due to Monte Carlo sampling (variance), and the second term to the error due to the linearization (bias). Thus $\bDelta$ can be seen as a tuning parameter which controls the trade-off between the two sources of error. This is apparent in the two limiting cases: for $\bDelta\to\bzero$ there is pure Monte Carlo sampling (no linearization) and for $\bDelta=\bSigma$ there is pure linearization (no sampling).

Theorem \ref{MSE theorem} suggests a way for choosing $\bDelta$ by minimizing the right-hand side of Eq. \eqref{MSE bound}. This leads to the \textit{semi-definite program}
\begin{align}\label{autocov prob}
        &\min_{\bDelta} \Psi(\bDelta; \bmu, \bSigma, \bff), \\ \nonumber
        &\text{s.t.} \ \bDelta \succeq \bzero,  \ \bSigma-\bDelta \succeq \bzero,
\end{align}
where
\begin{equation} \label{autocov obj}
    \begin{split}
    \Psi(\bDelta;\bmu, \bSigma, \bff) &= \frac{1}{N} \Tr\big((\bSigma-\bDelta)\nabla\bff(\bmu)^T\nabla\bff(\bmu)\big) \\
          &+ \frac{1}{4} \sum_{i=1}^{d_y}\Tr\big(\bDelta \nabla^2\bff_i(\bmu)\big)^2.
    \end{split} 
\end{equation}

For affine functions $\bff$, the second term tends to zero and the minimizer of Eq. \eqref{autocov prob} will be $\bDelta = \bSigma$, leading to a purely linear approximation of the mean. On the other hand, when the first term of Eq. \eqref{autocov obj} vanishes, the program will select a zero covariance, corresponding to a pure Monte Carlo approximation. In general, when both terms of Eq. \eqref{autocov obj} are active, the optimal covariance will balance linearization and Monte Carlo error to minimize the objective.
There are many optimization algorithms that can be used to solve Eq. \eqref{autocov prob}, such as projected gradient descent or interior point methods \cite{bertsekas1997nonlinear}. These algorithms however introduce additional computational complexity to the filtering algorithm, which solves Eq. \eqref{autocov prob} at every iteration.

{In this paper, we consider cases where the degree of nonlinearity of the function $\bff$ is comparable in all directions of the state space. In such cases, it is not necessary to solve the full optimization problem, and we use Eq. \eqref{autocov prob} to define a simpler problem that has an analytical solution, thus reducing the complexity of the filter. We arrive at the simpler problem by assuming that the augmentation covariance is proportional the original covariance as shown in the following theorem.
}

\begin{theorem}\label{thm 2} The program of Eq. \eqref{autocov prob}, when $\bDelta$ is constrained to be proportional to the covariance $\bDelta = \rho \bSigma$, is solved by,
\begin{equation} \label{rho opt}
    \rho^* = \frac{2}{N} \frac{\Tr\big(\bSigma \nabla\bff(\bmu)^T\nabla\bff(\bmu)\big)}{\sum_{i=1}^{d_y}\Tr\big(\bSigma\nabla^2\bff_i(\bmu)\big)^2} \wedge 1.
\end{equation}
\end{theorem}

\begin{proof}
    Substituting $\bDelta = \rho \bSigma$ into Eq. \eqref{autocov obj} we obtain the objective for $\rho$:
    \begin{equation}
     \begin{split}   
        \Psi(\rho) &= \frac{1}{N} (1-\rho) \Tr\big(\bSigma \nabla\bff(\bmu)^T\nabla\bff(\bmu)\big) \\
          &+ \rho^2 \frac{1}{4} \sum_{i=1}^{d_y}\Tr\big(\bSigma \nabla^2\bff_i(\bmu)\big)^2,
        \end{split}   
    \end{equation}
which is a quadratic function of $\rho$. Setting $\Psi'(\rho) = 0$ and solving for $\rho$ it is straightforward to obtain the desired expression. Noting that the constraints of the optimization problem \eqref{autocov prob} are translated to $0\leq \rho \leq 1$, we obtain Eq. \eqref{rho opt}.
\end{proof}
In the context of our filtering framework presented in the next section we use Eq. \eqref{rho opt} to derive a novel adaptive filtering algorithm. {Note that Theorems \ref{MSE theorem} and \ref{thm 2} have not been proved for the case of the unscented estimator because of the intractability of the bias. In our experiments with the unscented AGSF we use Eq. \eqref{rho opt} and demonstrate its empirical success.}

{Finally, it of interest to solve Eq. \eqref{autocov prob} for the complete augmentation covariances when the function $\bff$ is highly nonlinear in some, but not all directions of the state space. In such scenarios, it becomes favorable to split the Gaussian components along the directions of high nonlinearity. This results in covariance matrices $\bDelta$ which are shortened in the directions of high nonlinearity and elongated in the directions of low nonlinearity.}

\section{Augmented Gaussian Sum Filter}\label{sec 4}
We now derive a novel class of filtering algorithms, called \textit{augmented Gaussian sum filters} (AGSF), by leveraging the augmented Gaussian approximations presented in Section \ref{sec 3.B}. Our novel framework uses the novel approximation at each time step in the prediction and update steps of the filtering procedure. We note that all algorithms presented in this section are derived for additive models for notational simplicity and can be generalized to non-additive models straightforwardly.

In Section \ref{sec 4.A}, we start by presenting a generic version of the algorithm, which uses the generic \textit{augmented Gaussian approximation} of Section \ref{sec 3.A}. We also particularize using the augmented linear and unscented approximations, and derive the linear-AGSF (L-AGSF) and unscented-AGSF (U-AGSF) algorithms respectively. In Section \ref{sec 4.B}, we prove that the novel AGSF algorithm interpolates between the well known \textit{Gaussian sum filter} (GSF) and \textit{bootstrap particle filter} (BPF), which are special cases of the novel algorithm. Finally, in Section \ref{sec 4.C} we 
discuss strategies for selecting the parameters of the AGSF algorithm.

\subsection{Framework description}\label{sec 4.A}

 \subsubsection{Generic AGSF}\label{sec 4.A.1}
The new filter is summarized in Alg. ~\ref{AGSF} in its basic form. The prior is initialized as a mixture of $M$ Gaussians and iterates over prediction and update steps, whose derivations are given below. We assume that the augmentation covariance matrices $\{\bDelta^{(t)}_m\}_{m=1}^M$ and $\{\bLambda^{(t)}_{mn}\}_{m,n=1}^{M,N_m}$ are parameters of the algorithm set by the user at run time.

Also note that at each step, the number of components of the original mixture grows. To ensure that our algorithm will run on fixed memory and cost, we do a resampling step at the end of the update step (although other mechanisms are also possible).

\noindent\textbf{Prediction:}
Here we derive the predictive density our filter. We assume that the filtering distribution at time $t-1$ is given by the following Gaussian mixture of $M$ components
\begin{equation}
    p(\bx_{t-1}|\by_{1:t-1}) \simeq \sum_{m=1}^M w_{t-1}^{(m)}\mathcal{N}(\bx_{t-1}|\bmu_{t-1}^{(m)}, \bSigma_{t-1}^{(m)}).
\end{equation}
Substituting into Eq. \eqref{pred}, we obtain for the predictive distribution of $\bx_t$,
\begin{equation} \label{pred mix}
p(\bx_t|\by_{1:t-1}) \simeq \sum_{m=1}^M w^{(m)}_{t-1}  p_m(\bx_t{|\by_{1:t-1}}),
\end{equation}
where
\begin{equation*}
\resizebox{0.98\hsize}{!}{
  $p_m(\bx_t{|\by_{1:t-1}})= \int \mathcal{N}(\bx_{t}| \bff(\bx_{t-1}), \bQ)  \mathcal{N}(\bx_{t-1} | \bmu^{(m)}_{t-1}, \bSigma^{(m)}_{t-1}) d\bx_{t-1}$}.
\end{equation*}

This is the marginal pdf of $\bx_t$ with the joint pdf being
\begin{equation*}
\resizebox{0.95\hsize}{!}{
   $p_m(\bx_{t-1}, \bx_t{|\by_{1:t-1}}) = \mathcal{N}(\bx_{t}| \bff(\bx_{t-1}), \bQ)  \mathcal{N}(\bx_{t-1} | \bmu^{(m)}_{t-1}, \bSigma^{(m)}_{t-1})$}.
\end{equation*}

We will approximate the joint $p_m(\bx_{t-1}, \bx_t{|\by_{1:t-1}})$ by a Gaussian mixture using the \textit{augmented Gaussian approximation} of Section \ref{sec 3.B}. We will use the resulting marginal approximation for $p_m(\bx_t{|\by_{1:t-1}})$ which will also be given as a Gaussian mixture. Using Eqs. \eqref{mixture joint approx}-\eqref{cond cross-cov}, we obtain for the joint,
{
\begin{equation*}
p_m(\bx_{t-1}, \bx_t|\by_{1:t-1}) \simeq \frac{1}{N_m} \sum_{n=1}^{N_m} \mathcal{N}_{mn}(\bx_{t-1}, \bx_t),
\end{equation*}
where,
\begin{equation}
    \mathcal{N}_{mn}(\bx_{t-1}, \bx_t)  = \resizebox{0.7\hsize}{!}{$ \mathcal{N}\Bigg[ \begin{pmatrix} \bx_{t-1} \\   \bx_{t} \end{pmatrix} \Bigg| \begin{pmatrix} \bz_{mn} \\  \bmu_{t, mn}^- \end{pmatrix}, \begin{pmatrix} \bDelta^{(t)}_m &   \bC_{mn}^- \\ (\bC^-_{mn})^{T} &   \bSigma_{t, mn}^{-} \end{pmatrix} \Bigg],$}
\end{equation}
and $\bmu_{t, mn}^-, \bSigma_{t, mn}^-$ are given in Eqs. \eqref{AGSF pred mean}-\eqref{AGSF pred cov}, and 
\begin{equation} \label{AGSF pred crosscov}
    \bC_{mn}^- = \mathcal{I}[(\bx_{t-1} - \bz_{mn})(\bff(\bx_{t-1}) - \bmu_{t, mn}^-)^T ; \bz_{mn}, \bDelta^{(t)}_m].
\end{equation}}
The marginal approximation is given by
\begin{equation}
p_m(\bx_t{|\by_{1:t-1}}) \simeq \frac{1}{N_m} \sum_{n=1}^{N_m} \mathcal{N}(\bx_{t} |\bmu_{t, mn}^-, \bSigma_{t, mn}^-).
\end{equation}
Hence from Eq. \eqref{pred mix} we obtain the predictive distribution given in Eq. \eqref{AGSF pred mix}.

\begin{figure}
    \centering
    \begin{tcolorbox}[boxrule=0.2pt]
        \includegraphics[width=0.99\linewidth]{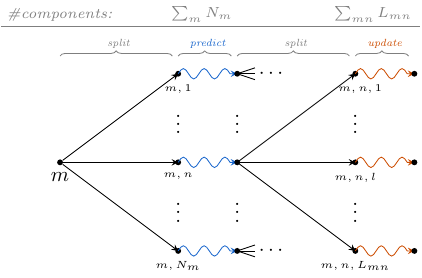}
    \end{tcolorbox}
    \caption{Splitting, prediction and update steps for a component of the AGSF.}
    \label{fig:AGSF step}
\end{figure}

\noindent\textbf{Update:} Here we show how the filtering pdf in Eq. \eqref{AGSF post mix} is obtained. Using the update equation, Eq. \eqref{upd}, with the prior given by Eq. \eqref{AGSF pred mix}, we approximate the posterior of $\bx_t$ as
\begin{equation} \label{joint_mix}
    \resizebox{0.95\hsize}{!}{
     $\widehat p(\bx_t|\by_{1:t}) \propto \sum_{m=1}^{M} \sum_{n=1}^{N_m} w_{mn} \mathcal{N}(\by_t|\bg(\bx_t), \bR)\mathcal{N}(\bx_{t}| \bmu_{t, mn}^-, \bSigma_{t, mn}^-)$}.
\end{equation}
As in the prediction step, we use the augmentation, to approximate the joint 
\begin{equation*}
    p_{mn}(\bx_t, \by_t{|\by_{1:t-1}}) = \mathcal{N}(\by_t|\bg(\bx_t), \bR)\mathcal{N}(\bx_{t}| \bmu_{t, mn}^-, \bSigma_{t, mn}^-)
\end{equation*}
by a Gaussian sum. Using Eq. \eqref{mixture joint approx} we obtain
\begin{equation}\label{joint GM 2}
p_{mn}(\bx_t, \by_t{|\by_{1:t-1}}) \simeq \frac{1}{L_{mn}} \sum_{\ell=1}^{L_{mn}} \mathcal{N}_{nm\ell}(\bx_t, \by_t),
\end{equation}
where,
\begin{equation}
     \mathcal{N}_{nm\ell}(\bx_t, \by_t) = \resizebox{0.6\hsize}{!}{$\mathcal{N}\Bigg[ \begin{pmatrix} \bx_t \\   \by_t \end{pmatrix} \Bigg| \begin{pmatrix} \bs_{mn\ell} \\   \bmu_{\by, mn\ell} \end{pmatrix}, \begin{pmatrix} \bLambda^{(t)}_{mn} &   \bC_{mn\ell} \\ \bC_{mn\ell}^T &   \bS_{\by, mn\ell} \end{pmatrix} \Bigg]$,}
\end{equation}
and $\bmu_{\by, mn\ell}, \bS_{\by, mn\ell}$ and $\bC_{mn\ell}$ are given in Eqs. \eqref{AGSF y-pred mean}-\eqref{AGSF y-pred crosscov}. 

{Combining Eqs. \eqref{joint_mix}-\eqref{joint GM 2} and applying  Eq. (13) from Appendix C (given in the supplementary material), we obtain Eq. \eqref{AGSF post mix} for the posterior of $\bx_t$.}

\noindent\textbf{Resampling:}
During the prediction and update steps of our algorithm, the number of Gaussian components is increased. We start with $M$ components before the prediction step to obtain $\sum_{m}N_m$ before and $\sum_{mn}L_{mn}$ after the update step, respectively. To ensure that our algorithm runs on a constant memory budget, we perform resampling at the end of the update step to reduce the number of components to $M$ although other mechanisms could be performed, {such as Gaussian mixture reduction \cite{schieferdecker2009gaussian, crouse2011look}}. An iteration of the generic AGSF algorithm for one of the mixture components is graphically illustrated in Fig. \ref{fig:AGSF step}

\subsubsection{Linear AGSF} The generic algorithm presented in Sec. \ref{sec 4.A.1} can not be used in practice since the expectations of Eqs. \eqref{AGSF pred mean}-\eqref{AGSF pred cov} and \eqref{AGSF y-pred mean}-\eqref{AGSF y-pred crosscov} need to be approximated using some numerical scheme. We derive the linear-AGSF (L-AGSF) algorithm by using the linear Taylor approximation of the nonlinearities $\bff$ and $\bg$ for Eqs.\eqref{AGSF pred mean}-\eqref{AGSF pred cov} and \eqref{AGSF y-pred mean}-\eqref{AGSF y-pred crosscov}. In this way, we obtain the L-AGSF given in Alg. 2 in the supplementary material.

\subsubsection{Unscented AGSF} We approximate the expectations in Eqs. \eqref{AGSF pred mean}-\eqref{AGSF pred cov} and \eqref{AGSF y-pred mean}-\eqref{AGSF y-pred crosscov} using the unscented approximation that we introduced in Sec. \ref{augm unscented approx}. In this way, we obtain the U-AGSF given in Algs. 3-4 the supplementary material.

\subsection{The AGSF unifies GSF and BPF}\label{sec 4.B}

The AGSF has the \textit{bootstrap particle filter} (BPF) and the \textit{Gaussian sum fitler} (GSF) as special cases. This is proved in the theorem that follows.

\begin{theorem}\label{theorem 3} The AGSF unifies the GSF and BPF in the following way:
    \begin{itemize}
        \item[\textit{(i)}] When at every timestep, the choices $N_m = 1, L_{mn}=1,$ and $\bDelta^{(t)}_m \to \bzero $, $\bLambda^{(t)}_{mn} \to \bzero$ are made for $m=1,\dots,M$, then Alg. \ref{AGSF} is identical to the BPF for an additive state-space model.
        \item[\textit{(ii)}] When at every timestep, the choices $\bDelta^{(t)}_m = \bSigma^{(m)}_{t-1}$, ${\bLambda^{(t)}_{mn} = \bSigma_{t, mn}^-}$ are made for all $m,n$, then the prediction and update steps of Alg. \ref{AGSF} are identical to those of a GSF for an additive state-space model.
    \end{itemize}
\end{theorem}

\begin{proof}
Given in Appendix B (see supplementary material).
\end{proof}
Theorem \ref{theorem 3} shows that the BPF and the GSF algorithms emerge as special cases in the AGSF framework for particular settings of the covariance parameters. The behavior of the AGSF for intermediate settings of the parameters is an interpolation between the BPF and the GSF which is determined by the choice of the augmentation covariances which is discussed in the next section. {Finally, we note that if in the definition of the augmented estimator of Eq. \eqref{A-KDE}, instead of vanilla Monte-Carlo we use importance sampling \cite{advancesinIS} with some proposal of our choice, the resulting AGSF algorithm will be and interpolation between a GSF and a particle filter which uses this particular proposal, this however is not proved in this paper.}

\subsection{Strategies for setting the parameters of AGSF} \label{sec 4.C}
The choice of the augmentation covariances plays an important role in the performance of the algorithms of the AGSF family. These covariances have to be set at every time step of the AGSF algorithms and they have to satisfy the constraints
\begin{align} \label{eq:agsf psd constr 1}
    \bDelta^{(t)}_m &\succeq \bzero,  \ \bSigma^{(m)}_{t-1}-\bDelta^{(t)}_m \succeq \bzero, \\ \label{eq:agsf psd constr 2}
    \bLambda_{mn}^{(t)} &\succeq \bzero,  \ \bSigma_{t, mn}^--\bLambda_{mn}^{(t)} \succeq \bzero,
\end{align}
for $m=1,\dots,M$, $n=1,\dots,N_m$ and $t=1,\dots,T$.

Moreover, the choices of $\bDelta^{(t)}_m$ and $\bLambda_{mn}^{(t)}$, at time step $t$ will determine whether the prediction and update steps, respectively, of the AGSF will resemble that of the BPF or the GSF, or {if the AGSF operates in a novel, in between setting}. When $\bDelta^{(t)}_m \simeq \bzero$ and $\bLambda_{mn}^{(t)} \simeq \bzero$ the behavior will be more like that of a BPF, whereas when $\bDelta^{(t)}_m \simeq \bSigma^{(m)}_{t-1}$ and $\bLambda_{mn}^{(t)} \simeq \bSigma_{t, mn}^-$ more like that of a GSF. {When $\bDelta^{(t)}_m$ and $\bLambda_{mn}^{(t)}$ are chosen in some intermediate values, i.e., not close to $\bzero$ or $\bSigma_{t-1}^{(m)}$ and $\bSigma_{t, mn}^-$ respectively, then the behavior of the AGSF will be in between the two.}

\begin{center}
\scalebox{1}{
\begin{minipage}{\linewidth}
\begin{algorithm}[H]
\begin{algorithmic}[1]
\STATE \textbf{Parameters:} $N_m$, $\bDelta^{(t)}_{m}$, and $L_{mn}$, $\bLambda^{(t)}_{mn}$ \ for $m=1,\dots,M$ $n=1,\dots,N_m$ and $t=1,\dots,T$
\STATE \textbf{Initialization:}
\STATE $\{w^{(m)}_{0}, \bmu^{(m)}_{0}, \bSigma^{(m)}_{0}\}_{m=1}^M$
\STATE for $t=1,\dots,T$:
\STATE \textbf{ Prediction:} \begin{align} \label{AGSF pred mix}
       \widehat p(\bx_{t}|\by_{1:t-1}) = \sum_{m=1}^M \sum_{n=1}^{N_m} w_{mn}  \mathcal{N}(\bx_{t}| \bmu_{t, mn}^-, \bSigma_{t, mn}^-),
\end{align}
where
\begin{small}
\begin{align} \label{AGSF pred mean}
    \bmu_{t, mn}^- &= \mathcal{I}[\bff(\bx_{t-1}) ; \bz_{mn}],\\ \label{AGSF pred cov}
    \bSigma_{t, mn}^-  &= \mathcal{I}[(\bff(\bx_{t-1}) - \bmu_{t, mn}^-)(\bff(\bx_{t-1}) - \bmu_{t, mn}^-)^T ; \bz_{mn}] + \bQ,
\end{align}
\end{small}
for $\bz_{mn} \sim \mathcal{N}(\bz|\bmu_{t-1}^{(m)}, \bSigma_{t-1}^{(m)}-\bDelta^{(t)}_m), w_{mn} = w_{t-1}^{(m)}/N_m$, $n=1,\dots,N_m$; $m=1,\dots,M$.
\STATE \textbf{ Update:}
\begin{align}\label{AGSF post mix}
    \widehat p(\bx_t|\by_{1:t}) &= \sum_{{m}=1}^M \sum_{n=1}^{N_m} \sum_{\ell=1}^{L_{mn}} w_{mnl} \mathcal{N}(\bx_t|\bmu_{{mn}\ell}, \bSigma_{{mn}\ell}),
\end{align}
where
\begin{align} \label{AGSF weights}
    w_{mnl} &\propto (w_{mn} / L_{mn})  \mathcal{N}(\by_t|\bmu_{\by, {mn}\ell}, \bS_{\by, {mn}\ell}), \\ \label{AGSF posterior mean}
    \bmu_{{mn}\ell} &= \bs_{{mn}\ell} + \bG_{{mn}\ell} (\by_t - \bmu_{\by, {mn}\ell}), \\ \label{AGSF posterior cov}
    \bSigma_{{mn}\ell} &= \bLambda^{(t)}_{mn} - \bG_{{mn}\ell} \bS_{\by, {mn}\ell} \bG_{{mn}\ell}^T,
\end{align}
with,
{
\begin{equation} \label{AGSF y-pred mean}
\hspace{-4.1cm}\resizebox{0.45\hsize}{!}{
$\bmu_{\by, mn\ell} = \mathcal{I}[\bg(\bx_{t}) ; \bs_{mn\ell}, \bLambda^{(t)}_{mn}]$},
\end{equation}
\begin{equation} \label{AGSF y-pred cov}
\resizebox{0.9\hsize}{!}{
$\bS_{\by, mn\ell} = \mathcal{I}[(\bg(\bx_t) - \bmu_{\by,mn\ell})(\bg(\bx_t) - \bmu_{\by, mn\ell})^T ; \bs_{mn\ell}, \bLambda^{(t)}_{mn}] + \bR$},
\end{equation}
}
and,
\begin{equation} \label{AGSF y-pred crosscov}
    \resizebox{0.9\hsize}{!}{
    $\bC_{mn\ell} = \mathcal{I}[(\bx_t - \bs_{mn\ell})(\bg(\bx_t) - \bmu_{\by, mn\ell})^T ; \bs_{mn\ell}, \bLambda^{(t)}_{mn}]$}, 
\end{equation}

\begin{equation} \label{AGSF gain}
\hspace{-4.4cm}\bG_{{mn}\ell} = \bC_{mn\ell} \bS_{\by, {mn}\ell}^{-1},
\end{equation}

\noindent where $\bs_{{mn}\ell} \sim \mathcal{N}(\bs | \bmu_{t, mn}^-, \bSigma_{t, mn}^- - \bLambda^{(t)}_{mn})$ for $m=1,\dots,M, n=1,\dots,N_m$ and $\ell = 1,\dots, L_{mn}$.
\STATE \textbf{ Resampling:} For $m'=1,\dots,M$, sample a triplet $(m n\ell)$ with probability $w_{mn\ell}$ and set 
\begin{equation*}
    \bmu^{(m')}_t = \bmu_{m n\ell} \ ; \ \bSigma_t^{(m')} = \bSigma_{m n\ell}.
\end{equation*}
Set $w_t^{(m)} = 1 / M$.
\end{algorithmic}
\caption{Augmented Gaussian sum filter}
\label{AGSF}
\end{algorithm}
\end{minipage}%
}
\end{center}

Hence it is reasonable to adopt the following heuristic. 
\textit{At time step $t$:}
\begin{itemize}
    \item[1.] \textit{if the function $\bff$ is highly nonlinear we choose $\bDelta^{(t)}_m$ close to $\bzero$. {If $\bff$ is almost linear we choose $\bDelta^{(t)}_m$ close to $\bSigma^{(m)}_{t-1}$, otherwise choose an intermediate $\bDelta^{(t)}_m$.}}
    \item[2.] \textit{if the function $\bg$ is highly nonlinear we choose $\bLambda_{mn}^{(t)}$ close to $\bzero$. {If $\bg$ is almost linear we choose $\bLambda_{mn}^{(t)}$ close to $\bSigma^{-}_{t,mn}$, otherwise choose an intermediate $\bLambda_{mn}^{(t)}$}}.
\end{itemize}
Note that this intuition behind the heuristic is captured in the optimization problem of Eq. \eqref{autocov prob}, where the quantification of the nonlinearity of $\bff$ and $\bg$ is done through the Hessian.

In our experiments, we use two different strategies to choose the covariances:
\begin{itemize}
    \item \textit{Proportional}: $\bDelta^{(t)}_m = \rho_1 \bSigma^{(m)}_{t-1}$ for some $0\leq\rho_1\leq1$ and $\bLambda_{mn}^{(t)} = \rho_2 \bSigma_{t, mn}^-$ for some $0\leq\rho_2\leq1$.
    \item \textit{Automatic:} We use the optimization problem in Eq. \eqref{autocov prob} or the approximate version of Eq. \eqref{rho opt} to set the covariances automatically. 
\end{itemize}
In both strategies the constraints in Eqs. \eqref{eq:agsf psd constr 1}-\eqref{eq:agsf psd constr 2} are satisfied. The two strategies can be used in any combination for choosing $\bDelta^{(t)}_m$ and $\bLambda_{mn}^{(t)}$. For example we may at each time-step choose $\bDelta^{(t)}_m$ proportionally, i.e., $\bDelta^{(t)}_m = \rho \bSigma^{(m)}_{t-1}$, while choosing $\bLambda_{mn}^{(t)}$ automatically, using the optimization problem. {Finally we note that }

\section{Numerical Experiments}\label{sec 5}
In this section, we present experimental results involving the novel AGSF algorithms as well as the classical GSF, BPF, {and APF algorithms.  
In Experiment A, we compare the performance the various algorithms for a classical maneuvering target tracking application (part of the results is in the supplementary material). In Experiment B, we evaluate the performance of an adaptive AGSF algorithm 
which selects the augmentation covariances automatically by exploiting Eq. \eqref{rho opt}.} All algorithms are implemented in JAX \cite{jax2018github} and the experiments are run on an Apple M2 Pro CPU.

\subsection*{Experiment A: Maneuvering target tracking}\label{section 5.A}
We consider a classical example from the signal processing literature, namely that of the tracking a maneuvering target in 2D, given measurements of its bearing with respect to the sensor and range, i.e., its distance from the sensor. We describe this with the following state-space model:
\begin{align}\label{BOT dyn}
    \bx_t &= \mathbf{F}_t\bx_{t-1} + \mathbf{G} \bq_t, \\ \label{BOT obs}
    \by_t &= \begin{pmatrix}
        \sqrt{x_{1t}^2 + x_{2t}^2} \\
        \tan^{-1}(x_{2t}/x_{1t})
    \end{pmatrix}   + \br_t,
\end{align}
where $\bx_t = (x_{1t}, v_{1t}, x_{2t}, v_{2t})$ is the 4-dimensional state vector, $(x_{1t}, x_{2t})$ is the position vector of the target and $(v_{1t}, v_{2t})$ the velocity vector. 

The motion of the maneuvering target is allowed to switch between the \textit{constant-velocity} (CV) and \textit{constant-turn} (CT) modes, described by the matrices, 
\begin{align}
    \mathbf{F}_{CV} &= \begin{pmatrix}
        1 & dt & 0 & 0 \\
        0 & 1 & 0 & 0 \\
        0 & 0 & 1 & dt \\
        0 & 0 & 0 & 1
    \end{pmatrix},
\end{align}
and
\begin{align}
    \mathbf{F}^{CT}_{\pm} &= \begin{pmatrix}
        1 & \sin(\Omega^{\pm}_t dt)\over\Omega^{\pm}_t & 0 & -\frac{1-\cos(\Omega^{\pm}_t dt)}{\Omega^{\pm}_t}  \\
        0 & \cos(\Omega^{\pm}_t dt)  & 0 & -\sin(\Omega^{\pm}_t dt) \\
        0 & \frac{1-\cos(\Omega^{\pm}_t dt)}{\Omega^{\pm}_t} & 1 & \sin(\Omega^{\pm}_t dt)\over\Omega^{\pm}_t \\
        0 & \sin(\Omega^{\pm}_t dt) & 0 & \cos(\Omega^{\pm}_t dt)
    \end{pmatrix},
\end{align}

\noindent where,
\begin{equation}
    \Omega^{\pm}_t = \pm \frac{a}{\sqrt{v_{1t}^2+v_{2t}^2}},
\end{equation}
$dt=1.0$. In this experiment we consider settings with $a=0.5$ and $a=0.05$. The dynamics of the target are defined by the matrix,
\begin{equation}
    \mathbf{F}_t = \begin{cases}
       \mathbf{F}^{CT}_+, & t \leq 2T/5, \\ 
       \mathbf{F}_{CV},  & 2T/5 < t \leq 3T/5\\ 
       \mathbf{F}^{CT}_-,  & 3T/5 < t\leq T,
    \end{cases},
\end{equation}

\noindent which describes an object following a clockwise circular motion, then a constant velocity motion, and finally an anti-clockwise circular motion. 

The dynamical noise vector $\bq_t$ is 2-dimensional with covariance matrix $\bQ = 10^{-6}\bI_2$. Finally,
\begin{equation}
 \ \mathbf{G} = \begin{pmatrix}
        0.5 & 1  \\
        1 & 0  \\
        0 & 0.5  \\
        0 & 1 
    \end{pmatrix},
\end{equation}
 
\noindent and $\bR = \sigma^2 \bI_2$, where $\sigma^2$ takes values in $\{25 \times 10^{-1}, 25 \times 10^{-3}, 25 \times 10^{-6}\}$. For the purposes of this experiment, all model parameters and switching behaviors are known by all filters.

\begin{table*}
\centering
\caption{{\textbf{{Experiment A}.} Comparison between the MSE and LPE for various settings of the L-GSF, U-GSF, and BPF with the proposed L-AGSF and U-AGSF algorithms. We compare the error metrics for $a=0.5$ and three observation noise values. Lower values are better.}}
\label{table::Exp A.1}
\resizebox{\textwidth}{!}{
\begin{tabular}{llcccccc}
\toprule
& & \multicolumn{2}{c}{$a=0.5, \sigma^2=25\times 10^{-1}$} & \multicolumn{2}{c}{$a=0.5, \sigma^2=25\times 10^{-3}$} & \multicolumn{2}{c}{$a=0.5, \sigma^2=25\times 10^{-6}$} \\
\cmidrule(lr){3-4} \cmidrule(lr){5-6} \cmidrule(lr){7-8}
& \textsc{parameters} & MSE & LPE & MSE & LPE & MSE & LPE \\
\midrule
\multirow{3}{*}{\rotatebox[origin=c]{90}{L-GSF}} & $M=1 (EKF)$ & $201.47\pm62.68$ & $3.07\times 10^{4}\pm4.16\times 10^{3}$ & $19.42\pm8.65$ & $9.13\times 10^{3}\pm2.18\times 10^{3}$ & $0.21\pm0.15$ & $9.42\times 10^{3}\pm6.17\times 10^{3}$ \\
  & $M=100$ & $62.48\pm13.46$ & $972.78\pm84.06$ (50.0\%) & $69.19\pm35.37$ & $3.07\times 10^{3}\pm2.56\times 10^{3}$ (95.0\%) & $0.44\pm0.43$ & $1.77\times 10^{4}\pm1.33\times 10^{4}$ \\
  & $M=1000$ & $101.22\pm21.27$ & $\textsc{n/a}$ (0\%) & $12.38\pm11.04$ & $92.67\pm8.42$ (90.0\%) & $0.02\pm6.56\times 10^{-3}$ & $1.08\times 10^{3}\pm138.07$ \\
\midrule
\multirow{3}{*}{\rotatebox[origin=c]{90}{U-GSF}} & $M=1 (UKF)$ & $15.37\pm8.60$ & $6.18\times 10^{5}\pm4.68\times 10^{5}$ & $4.13\pm1.97$ & $9.55\times 10^{3}\pm8.92\times 10^{3}$ & $0.02\pm5.39\times 10^{-3}$ & $713.35\pm293.59$ \\
  & $M=100$ & $16.40\pm8.57$ & $0.73\pm1.46$ & $0.25\pm0.13$ & $21.29\pm6.04$ & $0.06\pm0.04$ & $144.36\pm79.03$ \\
  & $M=1000$ & $22.07\pm9.20$ & $0.59\pm0.99$ & $0.40\pm0.15$ & $22.51\pm6.99$ & $7.10\times 10^{-3}\pm1.42\times 10^{-3}$ & $73.93\pm39.61$ \\
\midrule
\multirow{4}{*}{\rotatebox[origin=c]{90}{BPF}} & $M=100$ & $46.03\pm29.72$ & $4.55\times 10^{6}\pm2.49\times 10^{6}$ & $15.41\pm5.97$ & $6.69\times 10^{6}\pm2.85\times 10^{6}$ & $40.64\pm21.23$ & $1.81\times 10^{7}\pm1.09\times 10^{7}$ \\
  & $M=1000$ & $5.29\pm2.94$ & $3.99\times 10^{5}\pm3.14\times 10^{5}$ & $3.71\pm1.98$ & $1.42\times 10^{6}\pm9.45\times 10^{5}$ & $74.06\pm37.09$ & $3.82\times 10^{7}\pm2.15\times 10^{7}$ \\
  & $M=10000$ & $1.46\pm0.65$ & $3.98\times 10^{3}\pm3.92\times 10^{3}$ & $0.03\pm7.26\times 10^{-3}$ & $66.97\pm25.61$ & $74.21\pm33.05$ & $3.71\times 10^{7}\pm1.48\times 10^{7}$ \\
  & $M=100000$ & $0.76\pm0.20$ & $98.03\pm13.06$ & $0.03\pm5.06\times 10^{-3}$ & $21.76\pm10.65$ & $21.13\pm12.32$ & $1.10\times 10^{7}\pm6.84\times 10^{6}$ \\
\midrule
\multirow{4}{*}{\rotatebox[origin=c]{90}{APF}} & $M=100$ & $102.30\pm31.25$ & $2.77\times 10^{7}\pm8.34\times 10^{6}$ & $15.49\pm4.21$ & $6.38\times 10^{6}\pm1.94\times 10^{6}$ & $116.62\pm34.29$ & $5.94\times 10^{7}\pm1.59\times 10^{7}$ \\
  & $M=1000$ & $43.49\pm23.42$ & $1.12\times 10^{7}\pm6.06\times 10^{6}$ & $20.31\pm9.88$ & $8.62\times 10^{6}\pm3.96\times 10^{6}$ & $158.41\pm101.61$ & $6.97\times 10^{7}\pm4.64\times 10^{7}$ \\
  & $M=10000$ & $4.29\pm2.95$ & $9.81\times 10^{5}\pm1.05\times 10^{6}$ & $1.68\pm1.37$ & $6.48\times 10^{5}\pm4.86\times 10^{5}$ & $62.74\pm38.49$ & $3.45\times 10^{7}\pm2.10\times 10^{7}$ \\
  & $M=100000$ & $0.94\pm0.23$ & $349.70\pm90.77$ & $0.03\pm9.04\times 10^{-3}$ & $826.71\pm712.70$ & $0.94\pm1.03$ & $4.66\times 10^{5}\pm3.94\times 10^{5}$ \\
\midrule
\multirow{3}{*}{\rotatebox[origin=c]{90}{L-AGSF}} & $M=2, N=5, L=5$ & $75.62\pm19.18$ & $97.68\pm30.62$ & $1.29\pm0.67$ & $8.32\pm8.57$ & $0.06\pm9.83\times 10^{-3}$ & $-3.34\pm0.06$ \\
  & $M=10, N=5, L=5$ & $21.92\pm5.92$ & $17.36\pm3.44$ & $0.33\pm0.05$ & $-0.48\pm0.09$ & $0.02\pm1.40\times 10^{-3}$ & $-3.44\pm0.03$ \\
  & $M=100, N=5, L=5$ & $7.51\pm2.23$ & $5.25\pm0.45$ & $0.33\pm0.05$ & $-0.46\pm0.07$ & $0.01\pm1.67\times 10^{-3}$ & $-3.44\pm0.02$ \\
\midrule
\multirow{3}{*}{\rotatebox[origin=c]{90}{U-AGSF}} & $M=2, N=5, L=5$ & $38.42\pm13.29$ & $108.67\pm51.99$ & $0.50\pm0.14$ & $1.11\pm1.93$ & $0.02\pm2.26\times 10^{-3}$ & $-6.87\pm0.24$ \\
  & $M=10, N=5, L=5$ & $34.53\pm10.29$ & $72.97\pm34.87$ & $0.21\pm0.04$ & $-1.86\pm0.13$ & $0.01\pm1.27\times 10^{-3}$ & $-7.21\pm0.16$ \\
  & $M=100, N=5, L=5$ & $7.14\pm1.92$ & $4.52\pm0.95$ & $0.16\pm0.02$ & $-2.09\pm0.07$ & $0.01\pm1.08\times 10^{-3}$ & $-7.00\pm0.17$ \\
\bottomrule
\end{tabular}
}
\end{table*}

We compare the estimation performance for the various algorithms. To measure estimation performance, we use two metrics, the mean-squared error (MSE), and the log-probability error (LPE), both averaged across time that are defined below.

\begin{align}\label{mse}
    \text{MSE(alg)} &= \frac{1}{T}\sum_{t=1}^T ||\bx_{t}-\widehat \bx_{t}^{(\text{alg})}||^2 , \\ \label{lpe}
    \text{LPE(alg)} &= \frac{1}{T}\sum_{t=1}^T -\log \widehat p_{\text{alg}}(\bx_t|\by_{1:t-1})
\end{align}

\noindent where $\text{alg}\in\{\text{L-GSF, U-GSF, L-AGSF, U-AGSF, BPF, APF}\}$, $\widehat p_{\text{alg}}(\bx_t|\by_{1:t-1})$ is the filtered posterior estimate of the algorithm and $\widehat \bx_{t}^{(\text{alg})}$ the mean estimate. The MSE measures the accuracy of the point estimate (here, the posterior mean), i.e., how close \(\widehat{\bx}^{(\text{alg})}_{t}\) is to the true latent state \(\bx_{t}\) on average across time (we have access to the latent state only for evaluation purposes). In contrast, the LPE allows us to assess the quality of the entire filtered distribution \(\widehat p_{\text{alg}}(\bx_t\mid \by_{1:t-1})\), by penalizing algorithms that assign low probability density to the realised (true) state. This captures the behavior both in the point and uncertainty estimates of the algorithms.

Our results are shown in Tables \ref{table::Exp A.1} and III (given in the supplementary material). In each experiment, we consider different settings for each algorithm, varying the number of components or particles used. For each algorithm setting, we do $N_{\text{sim}}=100$ simulations of the state-space model in Eqs. \eqref{BOT dyn}-\eqref{BOT obs} for $T=500$ timesteps. Augmentation covariances for the prediction and update steps are chosen to be proportional to the filter covariance with $\rho_1=\rho_2 = 0.9$. Typical behaviors of L-GSF, L-AGSF, and BPF algorithms for various parameter settings are illustrated in Figs. 4-9, given in the supplementary material. We plot the position component $(x_1, x_2)$ of samples from the filtered posteriors for each algorithm. For the L-GSF and L-AGSF we also plot component means. In the rows we have an increasing number of components/particles for each algorithm ($M=100, M=500, M=1000$).

Across all scenarios, the PFs exhibit degeneracy leading to overconfident estimates, with the frequency and severity of this degeneracy depending strongly on both the observation noise and the particle number. In the higher- and moderate-noise settings ($\sigma^2=25\times 10^{-1}$ and $\sigma^2=25\times 10^{-3}$), this occurs for particle counts up to $M=10^4$, and its probability decreases as $M$ increases, as can be seen from the LPE metric. In the small-noise regime ($\sigma^2=25\times 10^{-6}$), degeneracy remains present even at $M=10^5$ with non-vanishing probability. This is due to the form of the weights of the BPF and APF, which are proportional to the observation distribution $p(\by_t|\bx_t)$. When the observation noise is very small the observation distribution is a very narrow Gaussian. Consequently, only the weights of particles that are closest to the observation-consistent manifold are non-vanishing, while the rest are essentially zero. As a result, after normalization and resampling, only a tiny subset of particles survives, the effective sample size collapses, and the filter becomes prone to locking onto an incorrect, overconfident mode, from which it cannot recover if the proposed particles lie far away from the true state. These behaviors are clearly illustrated in Figs. 4-9 of the supplementary material.
 
For the GSF family there is an opposite tendency: for higher noise settings ($\sigma^2=25\times 10^{-1}$ and  $\sigma^2=25\times 10^{-3}$), posterior components are often very spread out, and non-negligible weights are assigned to components with very low likelihood. Consequently, there is a lot of computational effort wasted on those components who do not contribute posterior mass. More crucially, when the weights of those low likelihood components become significant, the filter becomes unstable and produces posterior estimates that are not consistent with the true motion of the target, as can be witnessed in Figs. 4-5 and 7-8 of the supplementary material. As the noise decreases this behavior becomes less frequent, although a degree of instability exists also in the low noise setting ($\sigma^2=25\times 10^{-6}$), seen in Figs. 6 and 9 of the supplementary. 

Across scenarios, our results show that the AGSF behaves predictably and robustly. For higher observation noise ($\sigma^2=25\times 10^{-1}$ and $\sigma^2=25\times 10^{-3}$) the AGSF is able to represent a wider posterior uncertainty without wasting computational resources or becoming unstable. As the observation noise decreases, the AGSF is able to maintain a reasonable posterior estimate with reduced but not collapsed uncertainty or degeneracy, as can be seen from Figs. 6 and 9 of the supplementary. These results are corroborated by the MSE and LPE measurements in Tables \ref{table::Exp A.1} and III (given in the supplementary material), which show that the AGSF consistently avoids the pitfalls of both the BPF and GSF families, achieving good estimation accuracy and well-calibrated posteriors across all scenarios.
 
Overall, the AGSF is a very robust algorithm, using the augmentation covariances to avoid degeneracy on the one hand, while limiting overall posterior spread on the other. As a result, the AGSF reliably produces reasonable posterior estimates consistently. Finally, our experiments bolster the interpretation of the AGSF as a \textit{middle ground} between PFs and GSFs: avoiding degeneracy by using Gaussian components, which however are controllable and are not allowed to diverge. 

\subsection*{Experiment B: Adaptive AGSF}\label{section 5.B}
In this section, we present the adaptive version of the AGSF and evaluate its performance against other algorithms in a toy model. We show that by using Eq. \eqref{rho opt} the adaptive AGSF algorithm is able to change its behavior in real time to locally match the performance of a GSF or a BPF. This allows the adaptive AGSF to achieve better performance than both the GSF and BPF. 

To highlight the adaptive nature of the AGSF we choose to work with a SSM which switches between a linear Gaussian model and a stochastic volatility model, defined by
\begin{align} \label{eq:ssm2.1}
    \bx_t &= \bPhi \bx_{t-1} + \bq_t, \\  \label{eq:ssm2.2}
   \by_t &= u_t\bV_t \br_t + (1-u_t) (\bH \bx_t + \br_t),
\end{align}
where,
\begin{equation} \label{eq:ssm2.3}
    \bV_t = \beta \text{diag}(e^{\bx_{t,1}/\sigma}, \dots,e^{\bx_{t,d_x}/\sigma}),
\end{equation}
and $u_t\in[0,1]$ is a known input parameter. For $u_t=0$ the model is a linear Gaussian (LG) model whereas for $u_t=1$ it is a multivariate stochastic volatility (MSV) \cite{Chib2009}. We use the parameter values, $\bPhi = 0.8\bI$, $\bH = \bI$, $\beta = 0.5$, $\sigma = 4.0$, $\bQ = 10\bI$, and $\bR = 10^{-1}\bI$. The observation noise vector has mean $\br_0 = 10^{-4} \textbf{1}$ proportional to the all-ones vector. We set the dimensionality of the system to $d_x=d_y=4$ and use inputs that switch between 0 and 1 every 20 timesteps, starting from $u_1 = 0$. This defines a switching behavior that periodically shifts from LG to MSV model. The AGSF algorithms are run with proportional augmentation covariances in the the prediction step with $\rho_1 = 0.9$, i.e., $\bDelta^{(t)}_m = 0.9\bSigma^{(m)}_{t-1}$ for all $t$. In the update step the augmentation covariances are chosen proportional, with $\rho_2$ determined via Eq. \eqref{rho opt}.


\begin{table}
\centering
\caption{\textbf{Experiment B.}  {Comparison between the MSE and LPE of various settings of GSF, AGSF, and PF algorithms. Lower values are better.}}
\label{Table::Experiment B.1}
\resizebox{\columnwidth}{!}{
\begin{tabular}{llccc}
\toprule
& \textsc{parameters} & MSE & LPE & Runtime (s) \\
\midrule
\multirow{3}{*}{\rotatebox[origin=c]{90}{L-GSF}} & $M=10$ & $52.69\pm1.29$ & $6.61\pm0.04$ & $0.54\pm9.09\times 10^{-3}\,\mathrm{s}$ \\
  & $M=100$ & $53.21\pm2.59$ & $6.54\pm0.06$ & $0.63\pm0.01\,\mathrm{s}$ \\
  & $M=1000$ & $53.63\pm1.48$ & $6.64\pm0.03$ & $0.70\pm8.96\times 10^{-3}\,\mathrm{s}$ \\
\midrule
\multirow{3}{*}{\rotatebox[origin=c]{90}{U-GSF}} & $M=10$ & $49.90\pm1.36$ & $6.55\pm0.03$ & $0.69\pm0.03\,\mathrm{s}$ \\
  & $M=100$ & $52.27\pm2.12$ & $7.90\pm0.89$ & $0.86\pm0.08\,\mathrm{s}$ \\
  & $M=1000$ & $50.32\pm1.26$ & $6.55\pm0.03$ & $0.98\pm0.05\,\mathrm{s}$ \\
\midrule
\multirow{3}{*}{\rotatebox[origin=c]{90}{BPF}} & $M=100$ & $21.77\pm0.78$ & $2.07\times 10^{6}\pm8.56\times 10^{4}$ & $0.87\pm0.07\,\mathrm{s}$ \\
  & $M=1000$ & $15.77\pm0.29$ & $4.69\times 10^{5}\pm2.39\times 10^{4}$ & $0.87\pm0.04\,\mathrm{s}$ \\
  & $M=10000$ & $14.87\pm0.40$ & $1.20\times 10^{5}\pm5.03\times 10^{3}$ & $0.79\pm0.03\,\mathrm{s}$ \\
\midrule
\multirow{3}{*}{\rotatebox[origin=c]{90}{APF}} & $M=100$ & $20.79\pm0.39$ & $1.72\times 10^{6}\pm3.94\times 10^{4}$ & $0.92\pm0.07\,\mathrm{s}$ \\
  & $M=1000$ & $16.24\pm0.22$ & $3.33\times 10^{5}\pm1.42\times 10^{4}$ & $0.87\pm0.03\,\mathrm{s}$ \\
  & $M=10000$ & $15.44\pm0.39$ & $1.01\times 10^{5}\pm1.27\times 10^{4}$ & $0.87\pm0.02\,\mathrm{s}$ \\
\midrule
\multirow{3}{*}{\rotatebox[origin=c]{90}{L-AGSF}} & $M=10, N=5, L=5$ & $18.65\pm0.48$ & $5.65\pm0.03$ & $0.99\pm0.09\,\mathrm{s}$ \\
  & $M=100, N=5, L=5$ & $15.48\pm0.28$ & $5.47\pm0.02$ & $1.15\pm0.08\,\mathrm{s}$ \\
  & $M=1000, N=5, L=5$ & $14.59\pm0.26$ & $5.44\pm0.02$ & $1.08\pm0.05\,\mathrm{s}$ \\
\midrule
\multirow{3}{*}{\rotatebox[origin=c]{90}{U-AGSF}} & $M=10, N=5, L=5$ & $17.46\pm0.39$ & $5.58\pm0.02$ & $1.10\pm0.02\,\mathrm{s}$ \\
  & $M=100, N=5, L=5$ & $15.10\pm0.35$ & $5.46\pm0.02$ & $1.13\pm0.02\,\mathrm{s}$ \\
  & $M=1000, N=5, L=5$ & $14.46\pm0.32$ & $5.42\pm0.01$ & $1.23\pm0.01\,\mathrm{s}$ \\
\bottomrule
\end{tabular}}
\end{table}

\begin{figure}
     \centering
     \includegraphics[width=\columnwidth]{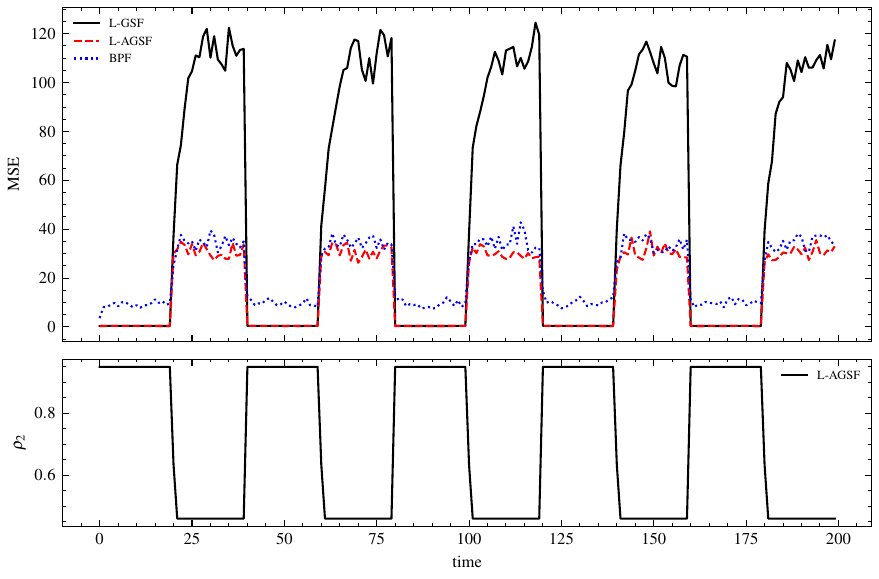}
    \caption{{\textbf{{Experiment B.}} \textit{(Top)} MSE plotted in time for L-GSF, U-GSF$(M=1000)$, L-AGSF, U-AGSF($M=100,N=5,L=5$) and BPF ($M=100$). \textit{(Bottom)} Plot of the proportionality parameters of the update step $\rho_2$. We see that when the model is linear the AGSF algorithms behave like the GSF and when it becomes nonlinear (MSV) the proportionality decreases, leading the AGSF algorithms to behave more like a particle filter.}} 
    \label{fig:exp B.1}
\end{figure}

We report our results for Experiment B in Table \ref{Table::Experiment B.1}. The results confirm our intuition which is that the adaptive AGSF, by switching its behavior in real-time will be able to outperform the GSF and BPF algorithm. The reason for the improved performance is illustrated in the top panel of Fig. \ref{fig:exp B.1}, which plots the MSE in time for a given setting of the L-GSF, L-AGSF, and BPF algorithms. As seen from the figure, the AGSF algorithm switches its behavior to match the best algorithm for the linear and nonlinear regimes. This allows it to achieve an optimal MSE in the course of the time-evolution of the model. In the bottom panel of Fig. \ref{fig:exp B.1}, we plot the coefficients $\rho_2$ in time. As we can see there, the coefficients are close to one when the model is linear and drop to approximately $0.5$ when the model becomes nonlinear. This nicely illustrates the mechanism of adaptation of the AGSF to the nonlinearity of the model. To the best of our knowledge, this type of adaptive behavior is new in the filtering literature. This experiment is a proof-of-concept for combining two filtering algorithms in this way, and the results are encouraging. We are excited about the potential applications that this approach can have in multiple fields where Bayesian filtering is routinely used. 

\section*{Conclusions}

We have proposed the augmented Gaussian sum filter, a novel class of algorithms which unify Gaussian sum and particle filters. The AGSF introduces an augmented Gaussian approximation, based on a set of latent variables and associated covariance parameters. By tuning the covariances it interpolates between GSFs and PFs which are special cases within the class. This allows for the design of a novel adaptive version of the algorithm which behaves more like a GSF or a PF according to the local nature of the nonlinearities. In a target-tracking application we have shown that the AGSF is an efficient algorithm which is robust towards common failure modes of GSFs and PFs. Moreover we have demonstrated experimentally the switching behavior of the adaptive AGSF. The AGSF framework can be extended in various ways and may be applied to many interesting problems. Future work could explore the use of more sophisticated proposal distributions beyond simple Monte Carlo sampling, such as variational or adaptive importance sampling schemes. Additionally, more advanced numerical integration methods could be incorporated to replace linearization or unscented transformations, enabling higher accuracy in capturing nonlinear dynamics. Finally, the adaptive AGSF’s ability to adjust to the local nonlinearities makes it well suited for detecting abrupt or gradual changes in system behavior. This opens the door to its use in tasks such as structural break detection, regime switching, and automatic model selection, where adaptive inference mechanisms are essential for identifying shifts in the underlying dynamics.

\appendices
\section*{Appendix A}\label{Appendix A}
\section*{Proof of Theorem 1}
The theorem is an application of Taylor's expansion to the MSE of the estimator of Eq. \eqref{mean estimator}. We write for the MSE,
\begin{align}\nonumber
    MSE(\widehat \bmu_{\by}) &= \mathbb{E}\big[||\widehat \bmu_{\by} - \bmu_{\by}||^2_2\big]  \\\nonumber
        & = \mathbb{E}\big[||\widehat \bmu_{\by} - \mathbb{E}[\widehat \bmu_{\by}] + \mathbb{E}[\widehat \bmu_{\by}] - \bmu_{\by}||^2_2\big] \\ \label{MSE1}
        & = \mathbb{E}\big[||\widehat \bmu_{\by} - \mathbb{E}[\widehat \bmu_{\by}] ||^2_2\big] + ||\mathbb{E}[\widehat \bmu_{\by}] - \bmu_{\by}||^2_2,
\end{align}
where the first term is the error due to Monte Carlo sampling, associated to the variance of the estimator, and the second term is the error due to linearization which corresponds to the bias. We now do a Taylor approximation in each term.

\noindent\textbf{Variance term:} We have
\begin{align}
    \mathbb{E}\big[||\widehat \bmu_{\by} -  \mathbb{E}\widehat \bmu_{\by}||_2^2\big] &=
    \mathbb{E}\big[\sum_{i=1}^{d_y} \big(\widehat \bmu_{\by,i} - \mathbb{E}\widehat \bmu_{\by,i}\big)^2\big] \\
    &= \sum_{i=1}^{d_y} \mathbb{V}ar\big[\widehat \bmu_{\by,i}\big].
\end{align}
For the variance of the scalar components of $\widehat \bmu_{\by}$ we have,
\begin{align}\label{variance eq}
    \mathbb{V}ar\big[\widehat \bmu_{\by,i}\big] = \mathbb{V}ar\big[\frac{1}{N}\sum_{n=1}^N \bff_i(\bz_n)\big] = \frac{\mathbb{V}ar\big[\bff_i(\bz)\big]}{N}.
\end{align}
The first order Taylor expansion of $\bff_i(\bz)$ around $\bmu$ is
\begin{equation*}
    \bff_i(\bz) \simeq \bff_i(\bmu) + \nabla\bff_i(\bmu)^T(\bz - \bmu).
\end{equation*}
Substituting this into Eq. \eqref{variance eq} we obtain
\begin{equation}
    \begin{split}
        \mathbb{V}ar\big[\widehat \bmu_{\by,i}\big] \simeq \frac{1}{N} \nabla\bff_i(\bmu)^T(\bSigma - \bDelta)\nabla\bff_i(\bmu).
    \end{split}
\end{equation}
Therefore we have for the variance term
\begin{equation}
    \begin{split}
        \mathbb{E}\big[||\widehat \bmu_{\by} -  \mathbb{E}\widehat \bmu_{\by}||_2^2\big] &\simeq \frac{1}{N}\sum_{i=1}^{d_y} \nabla\bff_i(\bmu)^T(\bSigma - \bDelta)\nabla\bff_i(\bmu) \\
        &= \frac{1}{N}\Tr\big( \nabla\bff(\bmu)^T\nabla\bff(\bmu)(\bSigma - \bDelta)\big).
    \end{split}
\end{equation}

\noindent\textbf{Bias term:} Similarly, the bias term is decomposed as
\begin{equation}\label{bias sum}
    ||\mathbb{E}[\widehat \bmu_{\by}] - \bmu_{\by}||^2_2 = \sum_{i=1}^{d_y} (\mathbb{E}\widehat \bmu_{\by,i} - \bmu_{\by,i})^2.
\end{equation}
We can approximate the terms in the sum by
\begin{align}
    \mathbb{E}[\widehat \bmu_{\by,i}] - \bmu_{\by,i} &= \mathbb{E}[\bff_i(\bz)] - \mathbb{E}\big[\mathbb{E}[\bff_i(\bx)|\bz]\big] \\ 
    &= \mathbb{E}\big[\bff_i(\bz) - \mathbb{E}[\bff_i(\bx)|\bz]\big] \\
    &\simeq -\frac{1}{2}  \Tr \big( \bDelta\mathbb{E}[\nabla^2\bff_i(\bz)] \big),
\end{align}
where we have used the second order Taylor expansion of $\bff_i(\bx)$ around $\bz$
\begin{equation*}
    \bff_i(\bx)  \simeq \bff_i(\bz) + \nabla\bff_i(\bz)^T(\bx - \bz) + \frac{1}{2} (\bx - \bz)^T\nabla^2\bff_i(\bz)(\bx - \bz),
\end{equation*}
and taken the expectation with respect to $p(\bx|\bz) = \mathcal{N}(\bx|\bz, \bDelta)$. Finally, by using the approximation $\mathbb{E}[\nabla^2\bff_i(\bz)] \simeq \nabla^2\bff_i(\bmu)$ and substituting into Eq. \eqref{bias sum}, we obtain for the bias term
\begin{equation}
    ||\mathbb{E}[\widehat \bmu_{\by}] - \bmu_{\by}||^2_2 \simeq \frac{1}{4}\sum_{i=1}^{d_y} \Tr \big(\bDelta \nabla^2\bff_i(\bmu) \big)^2.
\end{equation}

\bibliographystyle{IEEEbib_short}
{\small
\bibliography{biblio.bib}
}

\newpage


\renewcommand{\thetable}{\Roman{table}}
\setcounter{table}{2}
\setcounter{algorithm}{1}
\setcounter{figure}{3}

\begin{center}
{\large\textbf{Supplementary Material}}
\end{center}

\vspace{1em}

\section*{Appendix B} \label{Appendix B}
\section*{Proof of Theorem 3}
\textit{(i)} For an overview of the BPF algorithm we refer the reader to Algorithm 11.9 of Ref. [4] in the text. We initialize the AGSF recursion with a particle approximation (i.e. take ${\bSigma_{t-1}^{(m)} \to \bzero}, \ {m=1,\dots,M}$) and show that it is identical to that of the BPF.

\noindent \textbf{Prediction:} When $\bDelta^{(t)}_m = \bSigma_{t-1}^{(m)} \to \bzero$ we have ${\bz_{m} = \bmu_{t-1}^{(m)}, \ m=1,\dots,M}$, where we have dropped the $n$ index since we take $N_m=1$ for all $m$. Moreover, the expectations of Eqs. (60)-(61), are with respect to the delta measure, since
\begin{equation*}
    \mathcal{N}(\bx_{t-1}|\bz_{m}, \bDelta^{(t)}_m) \to \delta(\bx_{t-1} - \bz_{m}),
\end{equation*}
when $\bDelta_m^{(t)} \to \bzero$. Hence, we may substitute in these expressions $\bx_{t-1} = \bz_{m} = \bmu_{t-1}^{(m)}$ to obtain $\bmu_{m}^- = \bff(\bmu_{t-1}^{(m)})$ and $\bSigma_{m}^- = \bQ$ for all $m$.

\noindent \textbf{Update:} In the update step of the AGSF, we have (dropping indices $n, \ell$),  $\bs_m \sim \mathcal{N}(\bff(\bmu_{t-1}^{(m)}), \bQ)$, which is the pdf of the dynamical model $p(\cdot|\bx_{t-1} = \bmu_{t-1}^{(m)})$.

Moreover, in Eq. (68) the expectation is taken with respect to $\delta(\bx_t-\bs_m)$, since $\bLambda_m = \bzero$, hence $\bC_m=\bzero$ and the gain $\bG_m$ is also zero.
Hence, the update Eqs. (64)-(65) become, $\bmu_t^{(m)} = \bs_m \sim \mathcal{N}(\bff(\bmu_{t-1}^{(m)}), \bQ)$ and $\bSigma_t^{(m)} = \bzero$.
By identifying the means $\bmu_t^{(m)}$ with the particles of the BPF, we have ${\bmu_t^{(m)} \sim p(\cdot|\bmu_{t-1}^{(m)})}$ which is the sampling step of the BPF. Moreover, the weights of the particles are given by the likelihood $w_m \propto p(\by_t |\bmu_{t}^{(m)})$, just as in the BPF, as can be verified by Eqs.  (63), (66)-(67).
Finally, we have a resampling step with probabilities being proportional to the weights. Hence, we have reproduced a complete BPF recursion. {Note that the use of resampling is instrumental for the AGSF to capture the BPF as a special case.}

\noindent\textit{(ii)} For the $m$-th component of the AGSF, the choice ${\bDelta_m^{(t)} = \bSigma^{(m)}_{t-1}}$ implies that all sampled particles
\begin{equation*}
    \bz_{mn} \sim \mathcal{N}(\bmu_{t-1}^{(m)}, \bSigma_{t-1}^{(m)}-\bDelta_m^{(t)})
\end{equation*}
are identical to the mean $\bmu_{t-1}^{(m)}$. Hence, we can set $N_m=1$ for each $m$ without losing generality. Moreover from Eqs. (60)-(61) we can see that the expectations are taken with respect to the distribution
$\mathcal{N}(\bx_{t-1}|\bmu_{t-1}^{(m)}, \bSigma_{t-1}^{(m)})$, which are the original component distributions of the GSF.

Similarly for the update step, the sampled particles all coincide with the means of the predicted distributions $\bmu_{t, mn}^-$. Moreover once again, it is easy to see that the updated means and covariances are identical to those of a GSF.

\section*{Appendix C}\label{Appendix C}
\section*{Conditional of a Gaussian mixture}
Let the joint distribution of a pair of random vectors $\bx$ and $\by$ be given by the Gaussian mixture,
\begin{equation}\label{joint GM general}
p(\bx, \by) = \sum_{m=1}^{M} w_m \mathcal{N}_m(\bx,\by),
\end{equation}
where,
\begin{equation}
    \mathcal{N}_m(\bx,\by) = \mathcal{N}\Bigg[ \begin{pmatrix} \bx \\   \by \end{pmatrix} \Bigg| \begin{pmatrix} \bmu_{\bx, m} \\   \bmu_{\by, m} \end{pmatrix}, \begin{pmatrix} \bSigma_{\bx,m} &   \bC_{m} \\ \bC_{m}^T &   \bSigma_{\by,m} \end{pmatrix} \Bigg],
\end{equation}
are the components of the mixture. Let us denote by $\mathcal{N}_m(\bx), \mathcal{N}_m(\by)$ and $\mathcal{N}_m(\bx|\by)$, the marginals of each component and the conditional of $\bx$ given $\by$, given by respectively,
\begin{align}
    \mathcal{N}_m(\bx) &= \mathcal{N}(\bx | \bmu_{\bx,m}, \bSigma_{\bx, m}), \\
    \mathcal{N}_m(\by) &= \mathcal{N}(\by | \bmu_{\by,m}, \bSigma_{\by, m}), \\
    \mathcal{N}_m(\bx|\by) &= \mathcal{N}(\bx | \bmu_{\bx|\by,m}, \bSigma_{\bx|\by,m}),
\end{align}
where,
\begin{align}
    \bmu_{\bx|\by,m} &= \bmu_{\bx,m} + \bC_m\bSigma_{\by,m}^{-1}(\by - \bmu_{\by,m}), \\
    \bSigma_{\bx|\by,m} &= \bSigma_{\bx,m} - \bC_m \bSigma_{\by,m}^{-1} \bC_m^T,
\end{align}
see Lemma A.3 in Ref. [4]. It is straightforward to see that the marginals of the mixture are given by,
\begin{align}
    p(\bx) &= \sum_{m=1}^{M} w_m \mathcal{N}_m(\bx), \\
    p(\by) &= \sum_{m=1}^{M} w_m \mathcal{N}_m(\bx).
\end{align}
For the conditional of $\bx$ given $\by$ we have,
\begin{align}
    p(\bx|\by) &= \frac{p(\bx,\by)}{p(\by)} \\
    &= \frac{\sum_{m=1}^{M} w_m \mathcal{N}_m(\bx,\by)}{\sum_{m=1}^{M} w_m \mathcal{N}_m(\by)} \\
    &= \frac{\sum_{m=1}^{M} w_m \mathcal{N}_m(\by)\mathcal{N}_m(\bx|\by)}{\sum_{m=1}^{M} w_m \mathcal{N}_m(\by)}  \\ \label{generic mixture conditonal}
    &= \sum_{m=1}^{M} \tilde w_m \mathcal{N}_m(\bx|\by)
\end{align}
where,
\begin{equation}
    \tilde w_m = \frac{w_m \mathcal{N}_m(\by)}{\sum_{m=1}^{M} w_m \mathcal{N}_m(\by)}.
\end{equation}

\begin{center}
\begin{minipage}{\linewidth}
\begin{algorithm}[H]
\begin{algorithmic}[1]
\STATE \textbf{Parameters:} $N_m$, $\bDelta^{(t)}_{m}$, and $L_{mn}$, $\bLambda^{(t)}_{mn}$ \ for $m=1,\dots,M$ $n=1,\dots,N_m$ and $t=1,\dots,T$ \\
\STATE \textbf{Initialization:}
$\{w^{(m)}_{0}, \bmu^{(m)}_{0}, \bSigma^{(m)}_{0}\}_{m=1}^M$ \\
\STATE for $t=1,\dots,T$:
\STATE \textbf{ Prediction:} \begin{align*}
       \widehat p(\bx_{t}|\by_{1:t-1}) = \sum_{m=1}^M \sum_{n=1}^{N_m} w_{mn}  \mathcal{N}(\bx_{t}| \bmu_{t, mn}^-, \bSigma_{t, mn}^-),
\end{align*}
where
\begin{align*}
\bmu_{t, mn}^- &= \bff(\bz_{mn}),\\
\bSigma_{t, mn}^- &= \nabla \bff(\bz_{mn}) \bDelta^{(t)}_m \nabla \bff(\bz_{mn})^T + \bQ,
\end{align*}
for $\bz_{mn} \sim \mathcal{N}(\bz|\bmu_{t-1}^{(m)}, \bSigma_{t-1}^{(m)}-\bDelta^{(t)}_m)$, $w_{mn} = w_{t-1}^{(m)}/N_m$, $n=1,\dots,N_m$; $m=1,\dots,M$.
\STATE \textbf{ Update:}
\begin{align*}
    \widehat p(\bx_t|\by_{1:t}) &= \sum_{{m}=1}^M \sum_{n=1}^{N_m} \sum_{\ell=1}^{L_{mn}} w_{mnl} \mathcal{N}(\bx_t|\bmu_{{mn}\ell}, \bSigma_{{mn}\ell}),
\end{align*}
where
\begin{align*}
    w_{mnl} &\propto (w_{mn} / L_{mn})  \mathcal{N}(\by_t|\bmu_{\by, {mn}\ell}, \bS_{\by, {mn}\ell}), \\
    \bmu_{{mn}\ell} &= \bs_{{mn}\ell} + \bG_{{mn}\ell} (\by_t - \bmu_{\by, {mn}\ell}), \\
    \bSigma_{{mn}\ell} &= \bLambda^{(t)}_{mn} - \bG_{{mn}\ell} \bS_{\by, {mn}\ell} \bG_{{mn}\ell}^T,
\end{align*}
with
\begin{align*}
\bmu_{\by, {mn}\ell} &=  \bg(\bs_{{mn}\ell}), \\
\bS_{\by, {mn}\ell} &= \nabla \bg(\bs_{{mn}\ell})  \bLambda^{(t)}_{mn} \nabla \bg(\bs_{{mn}\ell})^T + \bR, \\
\bG_{{mn}\ell} &= \bLambda^{(t)}_{mn} \nabla \bg(\bs_{{mn}\ell})^T \bS_{\by, {mn}\ell}^{-1},
\end{align*}
\noindent where $\bs_{{mn}\ell} \sim \mathcal{N}(\bs | \bmu_{t, mn}^-, \bSigma_{t, mn}^- - \bLambda^{(t)}_{mn})$ for $m=1,\dots,M, n=1,\dots,N_m$ and $\ell = 1,\dots, L_{mn}$.
\STATE \textbf{ Resampling:} For $m'=1,\dots,M$, sample a triplet $(m n\ell)$ with probability $w_{mn\ell}$ and set
\begin{equation*}
    \bmu^{(m')}_t = \bmu_{m n\ell} \ ; \ \bSigma_t^{(m')} = \bSigma_{m n\ell}.
\end{equation*}
Set $w_t^{(m)} = 1 / M$.
\end{algorithmic}
\caption{Linear AGSF}
\label{L-AGSF}
\end{algorithm}
\end{minipage}
\end{center}

\begin{center}
\scalebox{0.8}{
\begin{minipage}{\linewidth}
\begin{algorithm}[H]
\begin{algorithmic}[1]
\STATE \textbf{Parameters:} $N_m$, $\bDelta^{(t)}_{m}$, and $L_{mn}$, $\bLambda^{(t)}_{mn}$ \ for $m=1,\dots,M$ $n=1,\dots,N_m$ and $t=1,\dots,T$ \\
\STATE \textbf{Initialization:}
$\{w^{(m)}_{0}, \bmu^{(m)}_{0}, \bSigma^{(m)}_{0}\}_{m=1}^M$ \\
\STATE for $t=1,\dots,T$:
\STATE \textbf{ Prediction:} \begin{align*}
       \widehat p(\bx_{t}|\by_{1:t-1}) = \sum_{m=1}^M \sum_{n=1}^{N_m} w_{mn}  \mathcal{N}(\bx_{t}| \bmu_{t, mn}^-, \bSigma_{t, mn}^-),
\end{align*}
where
\begin{align*}
\bmu_{t, mn}^- &= \sum_{i=-d_x}^{d_x} \omega_i \bff(\bsigma_{mn}^{(i)}),\\
\bSigma_{t, mn}^- &= \sum_{i=-d_x}^{d_x} \omega_i (\bff(\bsigma_{mn}^{(i)}) - \bmu_{t, mn}^-)(\bff(\bsigma_{mn}^{(i)}) - \bmu_{t, mn}^-)^T + \bQ,
\end{align*}
and
\begin{align*}
\bsigma_{mn}^{(0)} &= \bz_{mn}, \\
\bsigma_{mn}^{(\pm i)} &= \bz_{mn} \pm \sqrt{d_x+\lambda} \cdot [\bDelta_m^{(t)}]^{1/2}_{\bullet i}, \ i=1,\dots,d_x,
\end{align*}
for $\bz_{mn} \sim \mathcal{N}(\bz|\bmu_{t-1}^{(m)}, \bSigma_{t-1}^{(m)}-\bDelta^{(t)}_m)$, $w_{mn} = w_{t-1}^{(m)}/N_m$, $n=1,\dots,N_m$; $m=1,\dots,M$.
\end{algorithmic}
\caption{Unscented AGSF prediction}
\label{U-AGSF-pred}
\end{algorithm}
\end{minipage}%
}
\end{center}


\begin{table*}
\centering
\caption{{\textbf{{Experiment A}.} Comparison between the MSE and LPE for various settings of the L-GSF, U-GSF, and BPF with the proposed L-AGSF and U-AGSF algorithms. We compare the error metrics for $a=0.05$ and three observation noise values. Lower values are better.}}
\label{table::Exp A.2}
\resizebox{\textwidth}{!}{
\begin{tabular}{llcccccc}
\toprule
& & \multicolumn{2}{c}{$a=0.05, \sigma^2=25\times 10^{-1}$} & \multicolumn{2}{c}{$a=0.05, \sigma^2=25\times 10^{-3}$} & \multicolumn{2}{c}{$a=0.05, \sigma^2=25\times 10^{-6}$} \\
\cmidrule(lr){3-4} \cmidrule(lr){5-6} \cmidrule(lr){7-8}
& \textsc{parameters} & MSE & LPE & MSE & LPE & MSE & LPE \\
\midrule
\multirow{3}{*}{\rotatebox[origin=c]{90}{L-GSF}} & $M=1 (EKF)$ & $178.20\pm46.16$ & $3.38\times 10^{3}\pm896.60$ & $2.20\pm1.90$ & $669.81\pm610.31$ & $0.03\pm0.02$ & $465.31\pm215.72$ \\
  & $M=100$ & $52.11\pm9.41$ & $25.12\pm3.75$ & $0.35\pm0.17$ & $15.34\pm8.24$ & $9.27\pm9.12$ & $4.42\times 10^{5}\pm4.69\times 10^{5}$ \\
  & $M=1000$ & $66.80\pm20.83$ & $4.73\pm1.10$ (95.0\%) & $0.30\pm0.18$ & $14.20\pm15.84$ & $1.21\times 10^{-3}\pm7.16\times 10^{-4}$ & $89.39\pm96.67$ \\
\midrule
\multirow{3}{*}{\rotatebox[origin=c]{90}{U-GSF}} & $M=1 (UKF)$ & $87.79\pm48.81$ & $1.80\times 10^{4}\pm1.60\times 10^{4}$ & $2.95\pm2.45$ & $1.00\times 10^{5}\pm9.38\times 10^{4}$ & $0.07\pm0.06$ & $567.01\pm409.18$ \\
  & $M=100$ & $6.21\pm4.63$ & $7.26\pm7.20$ & $0.52\pm0.33$ & $-6.95\pm0.62$ & $1.28\times 10^{-3}\pm4.92\times 10^{-4}$ & $-15.07\pm0.19$ \\
  & $M=1000$ & $44.40\pm20.59$ & $-1.59\pm0.57$ & $0.16\pm0.05$ & $-7.36\pm0.45$ & $1.24\times 10^{-3}\pm4.46\times 10^{-4}$ & $-15.37\pm0.14$ \\
\midrule
\multirow{4}{*}{\rotatebox[origin=c]{90}{BPF}} & $M=100$ & $42.78\pm21.52$ & $1.61\times 10^{6}\pm8.85\times 10^{5}$ & $57.56\pm35.13$ & $2.02\times 10^{7}\pm1.36\times 10^{7}$ & $124.04\pm47.64$ & $5.98\times 10^{7}\pm2.70\times 10^{7}$ \\
& $M=1000$ & $2.21\pm0.76$ & $1.99\times 10^{3}\pm1.29\times 10^{3}$ & $97.77\pm59.90$ & $4.20\times 10^{7}\pm2.61\times 10^{7}$ & $151.37\pm45.90$ & $8.16\times 10^{7}\pm2.66\times 10^{7}$ \\
 & $M=10000$ & $1.34\pm0.36$ & $1.29\times 10^{3}\pm1.24\times 10^{3}$ & $1.74\pm1.08$ & $4.53\times 10^{5}\pm3.82\times 10^{5}$ & $159.01\pm47.90$ & $7.47\times 10^{7}\pm2.10\times 10^{7}$ \\
  & $M=100000$ & $1.92\pm0.80$ & $90.31\pm27.90$ & $0.02\pm3.01\times 10^{-3}$ & $7.95\pm8.45$ & $98.94\pm32.24$ & $4.71\times 10^{7}\pm1.36\times 10^{7}$ \\
\midrule
\multirow{4}{*}{\rotatebox[origin=c]{90}{APF}} & $M=100$ & $45.40\pm13.02$ & $2.21\times 10^{6}\pm5.56\times 10^{5}$ & $70.76\pm32.19$ & $2.96\times 10^{7}\pm1.44\times 10^{7}$ & $164.74\pm30.83$ & $8.37\times 10^{7}\pm1.52\times 10^{7}$ \\
  & $M=1000$ & $115.32\pm92.06$ & $9.59\times 10^{6}\pm8.30\times 10^{6}$ & $18.87\pm17.53$ & $8.94\times 10^{6}\pm8.26\times 10^{6}$ & $159.32\pm34.55$ & $7.65\times 10^{7}\pm1.70\times 10^{7}$ \\
  & $M=10000$ & $1.13\pm0.50$ & $5.80\times 10^{3}\pm2.93\times 10^{3}$ & $7.49\pm6.74$ & $3.34\times 10^{6}\pm3.18\times 10^{6}$ & $24.77\pm12.20$ & $1.30\times 10^{7}\pm5.86\times 10^{6}$ \\
  & $M=100000$ & $0.78\pm0.24$ & $183.63\pm44.74$ & $0.02\pm3.31\times 10^{-3}$ & $11.56\pm11.28$ & $68.97\pm43.13$ & $3.61\times 10^{7}\pm2.35\times 10^{7}$ \\
\midrule
\multirow{3}{*}{\rotatebox[origin=c]{90}{L-AGSF}} & $M=2, N=5, L=5$ & $94.86\pm28.93$ & $86.91\pm36.54$ & $6.01\pm5.06$ & $58.43\pm65.61$ & $7.97\times 10^{-3}\pm4.17\times 10^{-4}$ & $-4.17\pm0.02$ \\
  & $M=10, N=5, L=5$ & $21.45\pm6.45$ & $13.16\pm6.14$ & $0.56\pm0.08$ & $-1.28\pm0.08$ & $3.60\times 10^{-3}\pm4.32\times 10^{-4}$ & $-4.19\pm0.02$ \\
  & $M=100, N=5, L=5$ & $9.88\pm1.57$ & $3.07\pm0.07$ & $0.50\pm0.07$ & $-1.37\pm0.10$ & $2.93\times 10^{-3}\pm3.77\times 10^{-4}$ & $-4.16\pm0.02$ \\
\midrule
\multirow{3}{*}{\rotatebox[origin=c]{90}{U-AGSF}} & $M=2, N=5, L=5$ & $54.71\pm15.30$ & $86.04\pm30.01$ & $16.88\pm16.00$ & $404.43\pm349.29$ & $9.21\times 10^{-3}\pm1.34\times 10^{-3}$ & $-7.49\pm0.17$ \\
  & $M=10, N=5, L=5$ & $33.79\pm8.33$ & $36.29\pm11.35$ & $0.38\pm0.07$ & $-2.12\pm0.09$ & $5.03\times 10^{-3}\pm4.49\times 10^{-4}$ & $-7.56\pm0.08$ \\
  & $M=100, N=5, L=5$ & $8.49\pm1.72$ & $2.52\pm0.21$ & $0.33\pm0.04$ & $-2.16\pm0.07$ & $2.89\times 10^{-3}\pm2.98\times 10^{-4}$ & $-7.62\pm0.10$ \\
\bottomrule
\end{tabular}
}
\end{table*}

\begin{center}
\scalebox{0.8}{
\begin{minipage}{\linewidth}
\begin{algorithm}[H]
\begin{algorithmic}[1]
\STATE \textbf{ Update:}
\begin{align*}
    \widehat p(\bx_t|\by_{1:t}) &= \sum_{{m}=1}^M \sum_{n=1}^{N_m} \sum_{\ell=1}^{L_{mn}} w_{mnl} \mathcal{N}(\bx_t|\bmu_{{mn}\ell}, \bSigma_{{mn}\ell}),
\end{align*}
where
\begin{align*}
    w_{mnl} &\propto (w_{mn} / L_{mn})  \mathcal{N}(\by_t|\bmu_{\by, {mn}\ell}, \bS_{\by, {mn}\ell}), \\
    \bmu_{{mn}\ell} &= \bs_{{mn}\ell} + \bG_{{mn}\ell} (\by_t - \bmu_{\by, {mn}\ell}), \\
    \bSigma_{{mn}\ell} &= \bLambda^{(t)}_{mn} - \bG_{{mn}\ell} \bS_{\by, {mn}\ell} \bG_{{mn}\ell}^T,
\end{align*}
with
\begin{align*}
\bmu_{\by, {mn}\ell} &= \sum_{i=-d_x}^{d_x} \omega_i \bg(\bsigma_{mn\ell}^{(i)}), \\
\bS_{\by, {mn}\ell} &= \sum_{i=-d_x}^{d_x} \omega_i (\bg(\bsigma_{mn\ell}^{(i)}) - \bmu_{\by, {mn}\ell})(\bg(\bsigma_{mn\ell}^{(i)}) - \bmu_{\by, {mn}\ell})^T + \bR, \\
\bC_{mn\ell} &= \sum_{i=-d_x}^{d_x} \tilde \omega_i (\bsigma_{mn\ell}^{(i)}-\bs_{{mn}\ell})(\bg(\bsigma_{mn\ell}^{(i)})- \bmu_{\by, {mn}\ell})^T, \\
\bG_{{mn}\ell} &= \bC_{mn\ell} \bS_{\by, {mn}\ell}^{-1},
\end{align*}
where,
\begin{align*}
\bsigma_{mn\ell}^{(0)} &= \bs_{mn\ell}, \\
\bsigma_{mn\ell}^{(\pm i)} &= \bs_{mn\ell} \pm \sqrt{d_x+\lambda} \cdot [\bLambda_{mn}^{(t)}]^{1/2}_{\bullet i}, \ i=1,\dots,d_x,
\end{align*}
and $\bs_{{mn}\ell} \sim \mathcal{N}(\bs | \bmu_{t, mn}^-, \bSigma_{t, mn}^- - \bLambda^{(t)}_{mn})$ for $m=1,\dots,M, n=1,\dots,N_m$ and $\ell = 1,\dots, L_{mn}$.
\STATE \textbf{ Resampling:} For $m'=1,\dots,M$, sample a triplet $(m n\ell)$ with probability $w_{mn\ell}$ and set
\begin{equation*}
    \bmu^{(m')}_t = \bmu_{m n\ell} \ ; \ \bSigma_t^{(m')} = \bSigma_{m n\ell}.
\end{equation*}
Set $w_t^{(m)} = 1 / M$.
\end{algorithmic}
\caption{Unscented AGSF update and resampling}
\label{U-AGSF-upd}
\end{algorithm}
\end{minipage}%
}
\end{center}

\begin{figure*}
    \centering
    \includegraphics[width=0.5\linewidth]{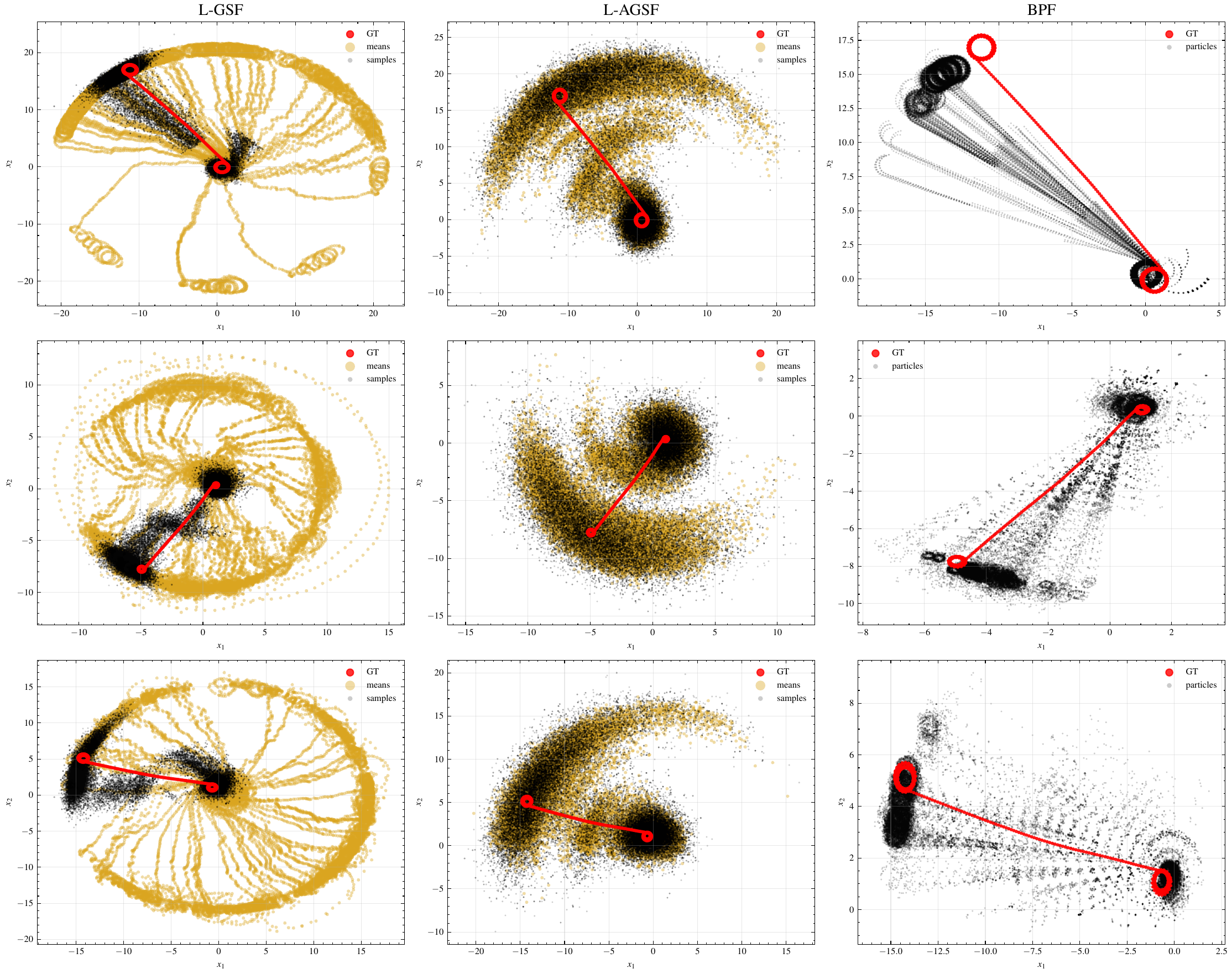}
    \caption{Plots of samples from the posteriors of different algorithms for the target tracking problem with $a=0.05$ and $\sigma^2 = 25\times 10^{-1}$. In each row we plot a different realization of the state trajectory, shared among all algorithms, with increasing number of components/particles ($M=100$, $M=500$, and $M=1000$ for the each row respectively). For GSF and AGSF algorithms we also plot component means. Typical posterior behaviors of the L‑GSF, L‑AGSF, and BPF are illustrated in this example. The true object trajectory is shown in red.}
    \label{fig:scatter_plot_0.05_2.5}
\end{figure*}

\begin{figure*}
    \centering
    \includegraphics[width=.5\linewidth]{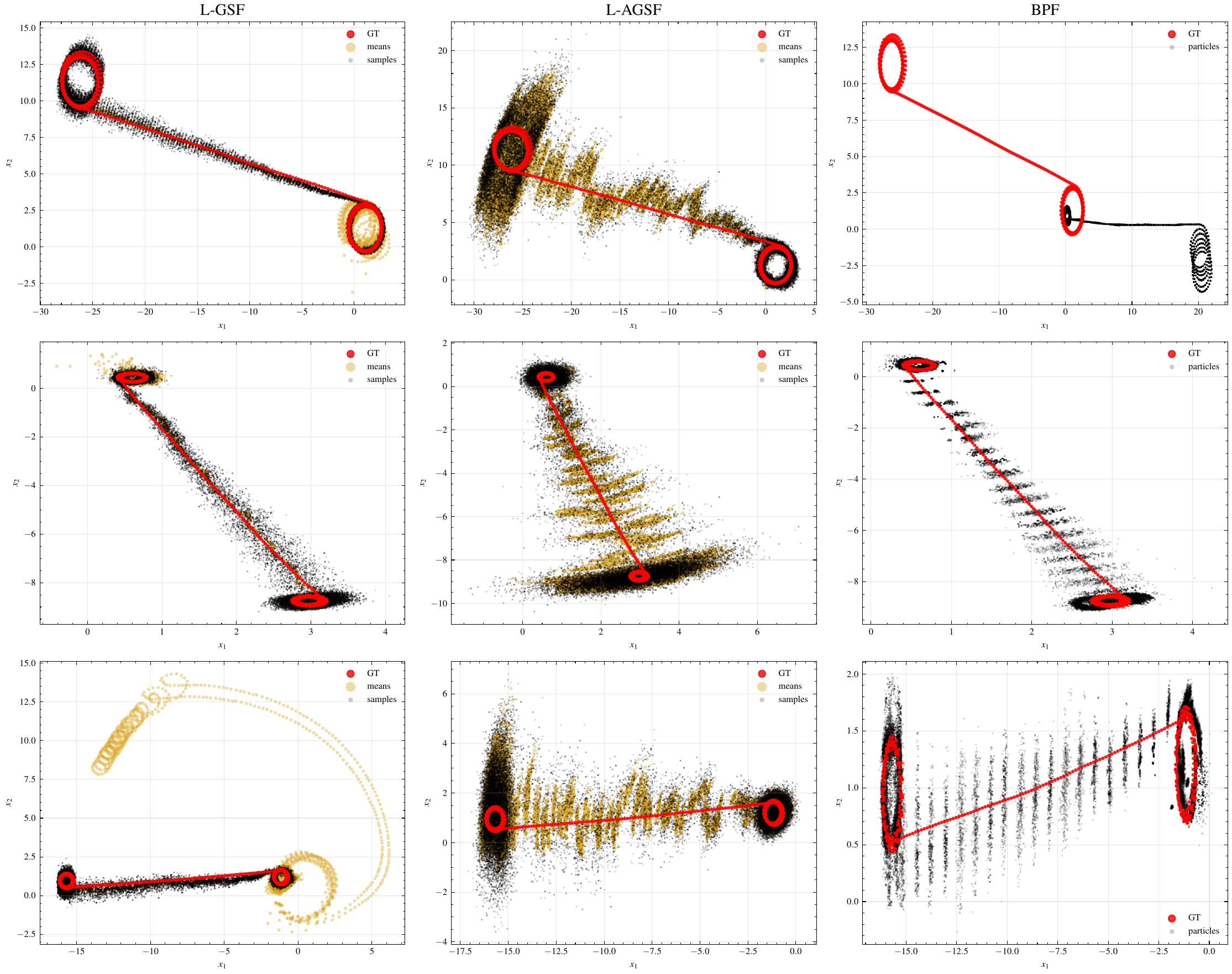}
    \caption{Plots of samples from the posteriors of different algorithms for the target tracking problem with $a=0.05$ and $\sigma^2 = 25\times 10^{-3}$. In each row we plot a different realization of the state trajectory, shared among all algorithms, with increasing number of components/particles ($M=100$, $M=500$, and $M=1000$ for the each row respectively). For GSF and AGSF algorithms we also plot component means. Typical posterior behaviors of the L‑GSF, L‑AGSF, and BPF are illustrated in this example. The true object trajectory is shown in red.}
    \label{fig:scatter_plot_0.05_0.025}
\end{figure*}

\begin{figure*}
    \centering
    \includegraphics[width=0.5\linewidth]{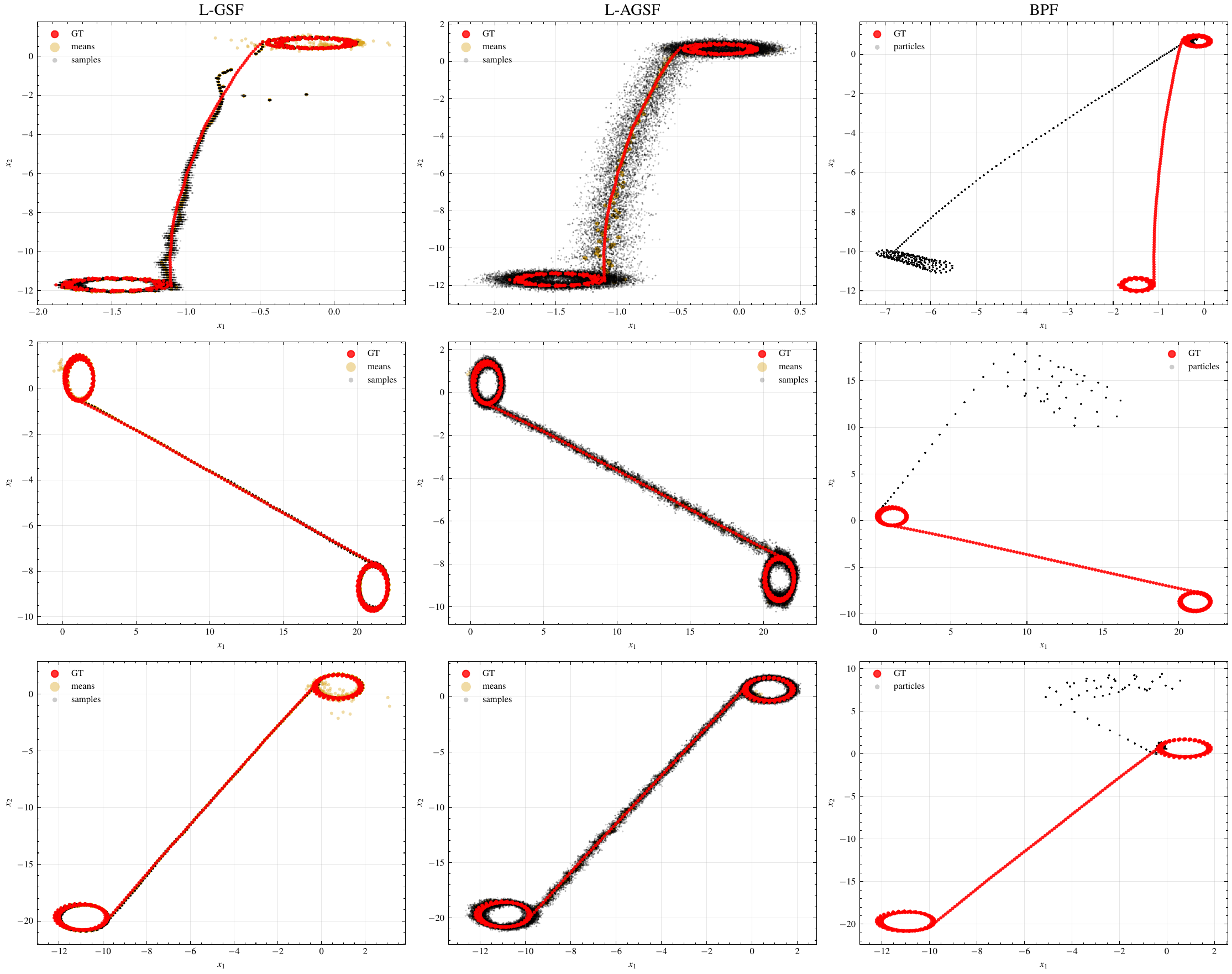}
    \caption{Plots of samples from the posteriors of different algorithms for the target tracking problem with $a=0.05$ and $\sigma^2 = 25\times 10^{-6}$. In each row we plot a different realization of the state trajectory, shared among all algorithms, with increasing number of components/particles ($M=100$, $M=500$, and $M=1000$ for the each row respectively). For GSF and AGSF algorithms we also plot component means. Typical posterior behaviors of the L‑GSF, L‑AGSF, and BPF are illustrated in this example. The true object trajectory is shown in red.}
    \label{fig:scatter_plot_0.05_0.000025}
\end{figure*}

\begin{figure*}
    \centering
    \includegraphics[width=0.5\linewidth]{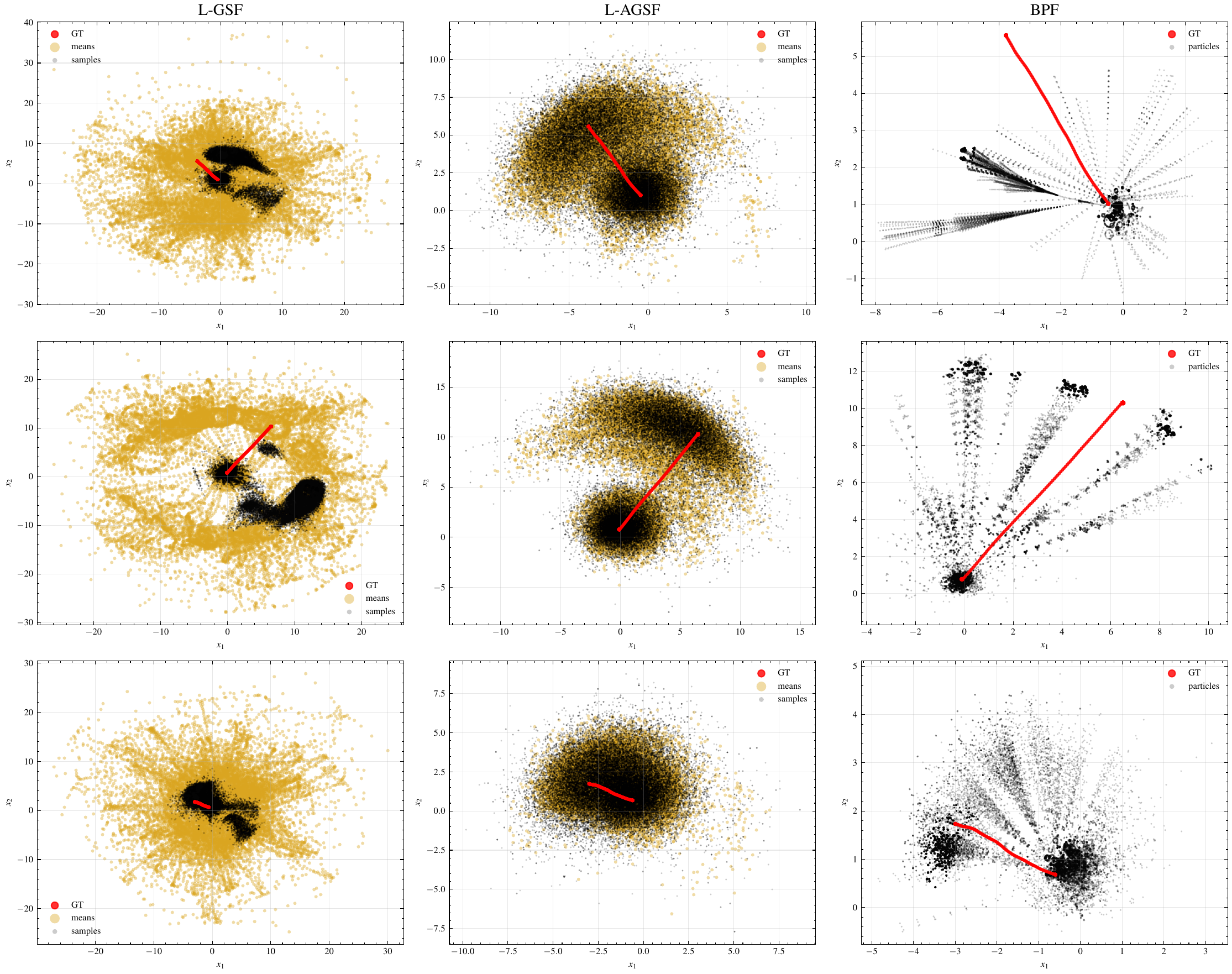}
    \caption{Plots of samples from the posteriors of different algorithms for the target tracking problem with $a=0.5$ and $\sigma^2 = 25\times 10^{-1}$. In each row we plot a different realization of the state trajectory, shared among all algorithms, with increasing number of components/particles ($M=100$, $M=500$, and $M=1000$ for the each row respectively). For GSF and AGSF algorithms we also plot component means. Typical posterior behaviors of the L‑GSF, L‑AGSF, and BPF are illustrated in this example. The true object trajectory is shown in red.}
    \label{fig:scatter_plot_0.5_2.5}
\end{figure*}

\begin{figure*}
    \centering
    \includegraphics[width=0.5\linewidth]{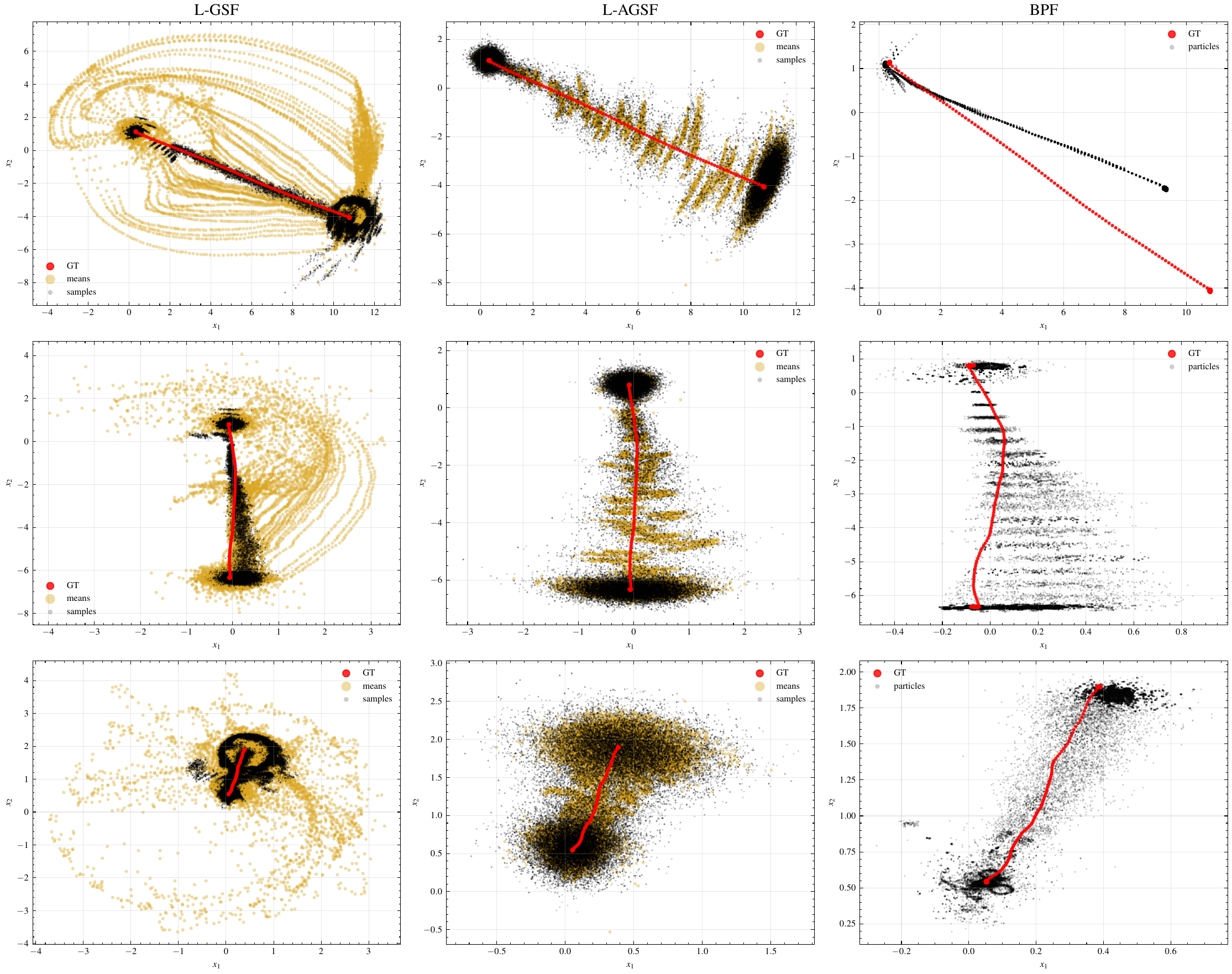}
    \caption{Plots of samples from the posteriors of different algorithms for the target tracking problem with $a=0.5$ and $\sigma^2 = 25\times 10^{-3}$. In each row we plot a different realization of the state trajectory, shared among all algorithms, with increasing number of components/particles ($M=100$, $M=500$, and $M=1000$ for the each row respectively). For GSF and AGSF algorithms we also plot component means. Typical posterior behaviors of the L‑GSF, L‑AGSF, and BPF are illustrated in this example. The true object trajectory is shown in red.}
    \label{fig:scatter_plot_0.5_0.025}
\end{figure*}

\begin{figure*}
    \centering
    \includegraphics[width=0.5\linewidth]{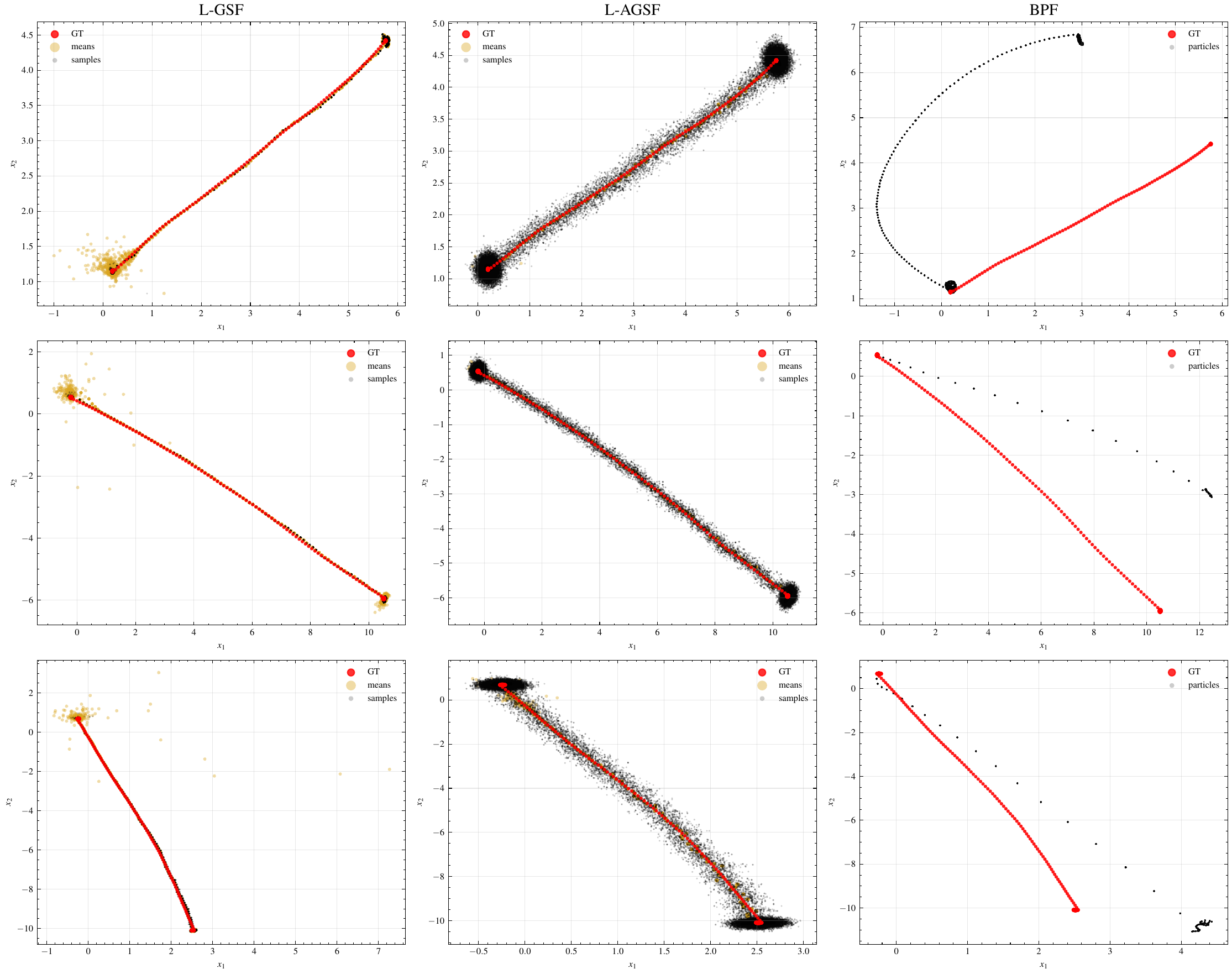}
    \caption{Plots of samples from the posteriors of different algorithms for the target tracking problem with $a=0.5$ and $\sigma^2 = 25\times 10^{-6}$. In each row we plot a different realization of the state trajectory, shared among all algorithms, with increasing number of components/particles ($M=100$, $M=500$, and $M=1000$ for the each row respectively). For GSF and AGSF algorithms we also plot component means. Typical posterior behaviors of the L‑GSF, L‑AGSF, and BPF are illustrated in this example. The true object trajectory is shown in red.}
    \label{fig:scatter_plot_0.5_0.000025}
\end{figure*}

\end{document}